\documentclass[12pt]{article}

\usepackage{amsmath,amsthm, amssymb, mathtools,stmaryrd}
\usepackage{type1cm}
\usepackage{mathrsfs}
\usepackage[T1]{fontenc}
\usepackage{fancyhdr}
\usepackage{appendix}
\usepackage{titlesec}
\usepackage{cite}
\usepackage{color}
\usepackage{physics}

\usepackage[top=25truemm,bottom=25truemm,left=20truemm,right=20truemm]{geometry}
\usepackage{graphicx}

\usepackage[thmtools-compat,store-sets-label]{keytheorems}

\usepackage[%
 setpagesize=false,%
 bookmarks=true,%
 bookmarksnumbered=true,%
 colorlinks=false,%
 pdftitle={},%
 pdfsubject={},%
 pdfauthor={},%
 pdfkeywords={}%
]{hyperref}

\baselineskip=\normalbaselineskip

\newcommand{\im}{\mathrm{i}}
\newcommand{\dif}{\mathrm{d}}
\newcommand{\e}{\mathrm{e}}

\newcommand{\C}{\mathbb{C}}

\usepackage{mathtools,slashed}

\newcommand\eq[1]{\begin{equation}#1\end{equation}}

\newcommand\als[1]{\begin{align}\begin{split}#1\end{split}\end{align}}

\makeatletter
 
 \@addtoreset{equation}{section}
\makeatother

\mathtoolsset{showonlyrefs=true}

\usepackage[all]{xy}

\usepackage{amscd}

\theoremstyle{definition} 
\newtheorem{definition}{Definition}
\newtheorem{assumption}{Assumption}

\theoremstyle{plain} 
\newtheorem{theorem}{Theorem}[section]

\newtheorem{lemma}[theorem]{Lemma}
\newtheorem{proposition}[theorem]{Proposition}

\theoremstyle{remark} 

\newtheoremstyle{underline}
{}        
{}              
{}              
{}    
{}              
{.}             
{1.5mm}         
{{\underline{\textit{\thmname{#1}\thmnumber{ #2}}~\thmnote{(#3)}\unskip}}}  

\theoremstyle{underline}

\begin{document}
\setcounter{footnote}{0}
\setcounter{tocdepth}{3}
\bigskip
\def\thefootnote{\arabic{footnote}}

\begin{titlepage}
\renewcommand{\thefootnote}{\fnsymbol{footnote}}
\begin{normalsize}
\begin{flushright}
\end{flushright}
  \end{normalsize}

~~\\

\vspace*{0cm}
    \begin{Large}
       \begin{center}
         {Spectral Codes: A Geometric Formalism for Quantum Error Correction}
       \end{center}
    \end{Large}
\vspace{0.7cm}

\begin{center}
Satoshi K\textsc{anno}\footnote[1]
            {
e-mail address : satoshi.kanno06@g.softbank.co.jp}
,
Yoshi-aki S\textsc{himada}\footnote[2]
            {
e-mail address : yoshiaki.shimada01@g.softbank.co.jp}

\vspace{0.7cm}

    {\it SoftBank Corp., Research Institute of Advanced Technology,}\\
               {\it 1-7-1 Kaigan, Minato-ku, Tokyo 105-7529,Japan}
               \end{center}

\vspace{0.5cm}

\begin{abstract}
\noindent
We present a new geometric perspective on quantum error correction based on spectral triples in noncommutative geometry. In this approach, quantum error-correcting codes are reformulated as low-energy spectral projections of Dirac-type operators that separate global logical degrees of freedom from local, correctable errors. Locality, code distance, and the Knill–Laflamme condition acquire a unified spectral and geometric interpretation in terms of the induced metric and spectrum of the Dirac operator.

Within this framework, a wide range of known error-correcting codes—including classical linear codes, stabilizer codes, GKP-type codes, and topological codes—are recovered from a single construction. This demonstrates that classical and quantum codes can be organized within a common geometric language.

A central advantage of the spectral-triple perspective is that the performance of error correction can be directly related to spectral properties. We show that leakage out of the code space is controlled by the spectral gap of the Dirac operator, and that code-preserving internal perturbations can systematically increase this gap without altering the encoded logical subspace. This yields a geometric mechanism for enhancing error-correction thresholds, which we illustrate explicitly for a stabilizer code.

We further interpret Berezin–Toeplitz quantization as a mixed spectral code and briefly discuss implications for holographic quantum error correction. Overall, our results suggest that quantum error correction can be viewed as a universal low-energy phenomenon governed by spectral geometry.

\end{abstract}

\end{titlepage}

\tableofcontents
\newpage

\section{Introduction}

A central principle of quantum error correction is the separation between global information and local errors. Logical information is encoded in degrees of freedom that cannot be accessed or modified by operations acting on sufficiently small subsystems, while physical noise typically acts locally. This separation underlies the robustness of encoded quantum information and is common to essentially all known quantum error-correcting codes.

In many concrete constructions, this principle is realized through a projection onto a low-energy subspace of a suitably chosen operator or Hamiltonian. Classical linear codes(see \cite{macwilliams1977theory}) arise as kernels of parity-check matrices, stabilizer codes(see \cite{gottesman1997stabilizer}) are defined as common eigenspaces of commuting operators, and both topological codes\cite{Kitaev:1997wr,dennis2002topological} and GKP-type codes\cite{gottesman2000encoding} appear as ground-state or low-energy sectors of physical Hamiltonians. Despite this shared structure, the relation between locality, low-energy projections, and the Knill--Laflamme error-correction condition has not been formulated within a unified mathematical framework.

The purpose of this work is to provide such a framework by viewing quantum error correction as a geometric phenomenon governed by spectral properties. Our central idea is to define code spaces as low-energy spectral projections of Dirac-type operators. In this perspective, global logical degrees of freedom correspond to modes that are insensitive to local geometric structure, while local errors are operators whose action is confined to finite regions defined by a spectral notion of distance.

To implement this idea, we employ the language of noncommutative geometry\cite{Landi:1997sh,Connes:1994yd,connes1985non,connes2000short}, where geometric information is encoded algebraically through spectral triples consisting of an algebra, a Hilbert space, and a Dirac operator. In this setting, the Dirac operator plays a dual role: it determines both a notion of locality via an induced metric and an energy scale via its spectrum. Low-energy modes of the Dirac operator naturally exhibit long-range correlations and therefore provide a natural home for globally encoded information.

Based on this observation, we introduce \emph{spectral codes}, defined as the zero-mode (or low-energy) subspaces of Dirac-type operators. Within this framework, locality is characterized by the spectral distance, and the Knill--Laflamme condition emerges as a geometric consequence of the separation between local and global degrees of freedom. This allows one to define code distance and correctable error sets in purely spectral terms.

A key strength of this approach is its unifying power. We show that a wide range of known error-correcting codes---including classical linear codes, stabilizer codes, GKP-type codes, and topological codes---arise naturally as spectral codes associated with appropriate choices of algebras and Dirac operators. In each case, the conventional notion of code distance is recovered from the same underlying spectral geometry.

Beyond unification, the spectral viewpoint provides new insight into error-correction performance. We show that leakage out of the code space is controlled by the spectral gap of the Dirac operator. Moreover, we demonstrate that \emph{code-preserving internal perturbations}---perturbations that leave the code space invariant while modifying higher-energy modes---can systematically increase this spectral gap. As a result, the error-correction threshold is enhanced without altering the encoded logical subspace. This mechanism is illustrated explicitly using a stabilizer code example.

We also relate spectral codes to strict deformation quantization, interpreting Berezin--Toeplitz quantization as a mixed spectral code in which quantum information is protected against classical geometric errors. Finally, we briefly discuss implications of the spectral code framework for holographic quantum error correction, where low-energy projections of noncommutative operator algebras naturally give rise to effective geometric descriptions.

\paragraph{Summary of contributions.}
The main contributions of this work are:
\begin{enumerate}
  \item The introduction of spectral codes, a geometric formulation of quantum error correction based on low-energy spectral projections.
  \item A geometric derivation of the Knill--Laflamme condition from spectral locality.
  \item A unified reconstruction of classical and quantum error-correcting codes within a single framework.
  \item A quantitative connection between spectral gaps and error-correction thresholds, together with a systematic method for threshold enhancement.
\end{enumerate}

The remainder of the paper is organized as follows. Section~2 reviews the elements of noncommutative geometry needed in this work. Section~3 introduces spectral codes and their basic properties. Section~4 analyzes decoding and threshold enhancement via internal perturbations. Section~5 discusses Berezin--Toeplitz quantization as a spectral code, and Section~6 explores connections to holography.

\paragraph{Related work.}

Connections between error-correcting codes and noncommutative geometry have been explored from several complementary perspectives.
Manin and Marcolli \cite{manin2009errorcorrectingcodesphasetransitions} related classical error-correcting codes to fractal geometry, operator algebras, and phase transitions.
Building on this direction, Bellissard, Marcolli, and Reihani \cite{bellissard2010dynamicalsystemsspectralmetric} constructed spectral triples associated with classical codes and related their spectral properties, including summability, to code parameters.
Marcolli and Perez \cite{marcolli2012codes} further developed noncommutative-geometric descriptions of classical and quantum stabilizer codes using code-associated fractals, noncommutative tori, and group actions, and discussed related spectral-triple constructions.

Taken together, these works reveal rich geometric and operator-algebraic structures associated with coding theory. The conceptual distinction of the present work lies in shifting the focus from the geometric structures associated with error-correcting codes to the geometric principles underlying error correction itself. Rather than primarily asking how a given code can be represented geometrically, we ask how the protection of information against local disturbances can be understood as a geometric phenomenon. Our central idea is that, within a suitable spectral triple, the Dirac operator characterizes both the low-energy sector carrying logical information and the notion of locality through its induced metric. From this perspective, quantum error correction can be understood, under appropriate conditions, as a manifestation of the geometric separation between global information and local disturbances.

\section{Foundations of Noncommutative Geometry}

This section develops the mathematical foundations required for a spectral formulation of quantum error-correcting codes.
In our approach, error-correcting code is described geometrically as a projection between Riemannian manifolds characterized by spectral triples, consisting of an algebra, a Hilbert space, and a Dirac operator.
The encoding map of an error-correcting code is identified with the projection onto the zero-mode (or low-energy) eigenspace of the Dirac operator.

To make this construction precise, we first review the notion of spectral triples and explain how they encod geometric information.
Since quantum systems are inherently noncommutative, the relevant geometric structures are described by noncommutative algebras rather than algebras of functions on ordinary manifolds.
Accordingly, we discuss how concepts such as space, distance, and locality are generalized in noncommutative geometry.

This section reviews the basic notions of noncommutative geometry needed in the sequel.
In 2.1, we introduce the definition of a noncommutative space and explain how the notion of a ``point'' can be formulated in this setting.
In 2.2, we discuss how to define distances on noncommutative spaces and, based on this, introduce the notion of a noncommutative Riemannian manifold.
In 2.3, we define what we call global information and local transformations in noncommutative geometry.


\subsection{Gelfand--Naimark Theorem and Noncommutative Spaces}

The starting point of noncommutative geometry is the idea that a space should be described not directly by a set of points, but by an algebra of functions defined on it. In this approach, algebraic structures are regarded as more fundamental than the underlying point-set topology. This perspective naturally leads to the notion of a $C^*$-algebra (see \cite{Landi:1997sh}) as the basic object encoding geometric information.

\begin{definition}[{$C^*$-algebra}]
    A complex algebra $A$ equipped with a product, a unit element $1\in A$, an antilinear involution $*:A\to A$, and a norm $\|\cdot\|$ is called a (unital) $C^*$-algebra if the following conditions hold:
    \begin{enumerate}
        \item $(ab)^* = b^* a^*$, $(a^*)^* = a$, and $1^* = 1$,
        \item $\|ab\| \le \|a\|\,\|b\|$ and $\|a^*\| = \|a\|$,
        \item the $C^*$-identity $\|a^*a\| = \|a\|^2$,
        \item $(A,\|\cdot\|)$ is complete, i.e.\ a Banach space.
    \end{enumerate}
\end{definition}

In the commutative case, $C^*$-algebras admit a direct geometric interpretation. This is formalized by the Gelfand--Naimark theorem, which establishes an equivalence between commutative $C^*$-algebras and compact Hausdorff spaces.

\begin{theorem}[Gelfand--Naimark \cite{1943OnTI}, commutative case]
For any compact Hausdorff space $X$, the algebra $C(X)$ of continuous complex-valued functions on $X$ is a commutative unital $C^*$-algebra. Conversely, any commutative unital $C^*$-algebra $A$ is $*$-isomorphic to $C(X)$ for some compact Hausdorff space $X$, uniquely determined up to homeomorphism.
\end{theorem}

This result allows one to reconstruct a space from its algebra of functions, establishing the principle that a commutative $C^*$-algebra may be regarded as the algebra of functions on a space. Noncommutative geometry arises by extending this principle to noncommutative $C^*$-algebras, where no underlying point-set space exists a priori.

\medskip

In general, a noncommutative $C^*$-algebra cannot be written in the form $C(X)$ for any topological space $X$. Nevertheless, noncommutative geometry interprets such an algebra itself as defining a ``noncommutative space.'' To make this interpretation precise, one must identify appropriate replacements for classical geometric notions such as points and topology.

In the commutative case, points of a space correspond to evaluation functionals. This observation motivates the use of states and pure states as generalized points in the noncommutative setting.

\begin{definition}[State]
    Let $A$ be a unital $C^*$-algebra. A linear functional $\omega:A\to\C$ is called a state if it satisfies
    \[
        \omega(a^*a) \ge 0 \quad \text{for all } a\in A, 
        \qquad \omega(1)=1.
    \]
    The set of all states on $A$ is denoted by $S(A)$.
\end{definition}

The state space $S(A)$ is a convex set. Physically, states correspond to expectation-value functionals, and convex combinations represent probabilistic mixtures of states. This physical interpretation motivates the distinguished role of extremal states.

\begin{definition}[Pure state]
    A state $\omega\in S(A)$ is called pure if it is an extreme point of the convex set $S(A)$. The set of pure states is denoted by $PS(A)$.
\end{definition}

In the commutative case $A=C(X)$, pure states are precisely point evaluations \cite{Landi:1997sh}:
\[
    \delta_x(f) = f(x), \qquad x\in X.
\]
As a consequence, one obtains a canonical identification
\[
    PS(C(X)) \cong X.
\]
Thus, for commutative $C^*$-algebras, points of the underlying space are recovered as pure states.

Motivated by this correspondence, noncommutative geometry interprets the pure state space $PS(A)$ of a general (possibly noncommutative) $C^*$-algebra as a substitute for the notion of points. While $PS(A)$ alone does not capture all geometric information in the noncommutative case, it provides a fundamental structure encoding probabilistic and topological aspects at the level of states of the noncommutative space.

\medskip

To recover a notion of topology, one equips the state space with a natural weak-$*$ topology \cite{Landi:1997sh}.

\begin{definition}[Gelfand topology]
Let $\{\omega_\lambda\}\subset S(A)$ be a net and $\omega\in S(A)$. We say that $\omega_\lambda$ converges to $\omega$ if
\[
    \omega_\lambda(a) \to \omega(a) \quad \text{for all } a\in A.
\]
This topology is the weak-$*$ topology induced from the dual space $A^*$ and is referred to as the weak-$*$ (or Gelfand) topology on the state space.
\end{definition}

In the commutative case $A=C(X)$, the Gelfand topology on $PS(A)$ coincides with the original topology on $X$, thereby recovering the full topological structure of the space. In the noncommutative case, the weak-$*$ topology on $S(A)$ and $PS(A)$ provides a natural notion of continuity and convergence, encoding the topological content of the corresponding noncommutative space.

This algebraic reconstruction of space captures its topological aspects. However, additional structures are required to describe geometric notions such as distance and metric properties. These will be introduced in the next subsection via spectral triples, which furnish a noncommutative generalization of Riemannian geometry.

\subsection{Connes distance and Noncommutative Riemannian Manifolds}

In the previous subsection, we explained how commutative $C^*$-algebras encode topological spaces and how noncommutative $C^*$-algebras may be interpreted as noncommutative spaces. However, topological data alone are insufficient to describe geometric notions such as distance, metric structure, or curvature. In noncommutative geometry, these additional geometric structures are incorporated by introducing an unbounded operator whose spectral properties encode metric information. This leads to the notion of a spectral triple \cite{Connes1992,connes2013spectral,Connes:1994yd}, which serves as the fundamental definition of a (possibly noncommutative) Riemannian manifold.

\begin{definition}[Spectral triple]
    A spectral triple $(A,\mathcal{H},D)$ consists of
    \begin{itemize}
        \item a $*$-algebra $A$,
        \item a Hilbert space $\mathcal{H}$,
        \item a faithful $*$-representation $\pi:A\to\mathcal{B}(\mathcal{H})$ (which we suppress in the notation),
        \item a densely defined self-adjoint operator $D$ on $\mathcal{H}$,
    \end{itemize}
    such that the following conditions are satisfied:
    \begin{enumerate}
        \item For all $a\in A$, the commutator $[D,a]$ extends to a bounded operator on $\mathcal{H}$.
        \item There exists a two-sided ideal $A_0\subset A$ such that for all $a_0\in A_0$,
        \[
            a_0 (D-\im)^{-1} \in \mathcal{K}(\mathcal{H}),
        \]
        where $\mathcal{K}(\mathcal{H})$ denotes the algebra of compact operators on $\mathcal{H}$.
    \end{enumerate}
\end{definition}

If $A$ is unital, the second condition is often replaced by the requirement that $(D-\im)^{-1}$ itself is a compact operator. This means that $D$ has compact resolvent, or equivalently, that its spectrum consists of discrete eigenvalues with finite multiplicities accumulating only at infinity. Geometrically, this condition corresponds to closedness of the underlying space.

A spectral triple simultaneously encodes the topological data of a space through the algebra $A$ and its geometric data through the Dirac operator $D$. For this reason, spectral triples are interpreted as noncommutative Riemannian manifolds.

\medskip

A key feature of spectral triples is that they provide an intrinsic notion of distance. Given a spectral triple $(A,\mathcal{H},D)$, Connes introduced a distance function \cite{Connes1992,connes2013spectral,Connes:1994yd} on the state space of $A$.

\begin{definition}[Connes distance]
    Let $(A,\mathcal{H},D)$ be a spectral triple. For two states $\rho,\omega\in S(A)$, the Connes distance is defined by
    \[
        d_D(\rho,\omega)
        := \sup\Bigl\{\,|\rho(a)-\omega(a)| \;\Big|\; a\in A,\ \|[D,a]\|\le 1 \Bigr\}.
    \]
    When restricted to pure states $\rho,\omega\in PS(A)$, this formula recovers a generalized notion of distance between points.
\end{definition}

In the commutative case, this definition reproduces the classical Riemannian distance. Let $M$ be a closed Riemannian spin manifold, $S\to M$ its spinor bundle, and $\slashed{D}$ the Dirac operator acting on $L^2(S)$. Then
\[
    \bigl(C^\infty(M),\, L^2(S),\, \slashed{D}\bigr)
\]
defines a spectral triple, and the Connes distance satisfies (see \cite{Landi:1997sh})
\[
    d_{\slashed{D}}(\delta_x,\delta_y)=d_M(x,y),
\]
where $d_M$ denotes the geodesic distance on $M$. Thus, the spectral triple framework extends the Gelfand--Naimark correspondence from topology to Riemannian geometry.

\medskip

The relationship between spectral triples and classical geometry is formalized by Connes' reconstruction theorem \cite{connes2013spectral}. Roughly speaking, this theorem states that any commutative spectral triple satisfying a suitable list of axioms (regularity, finiteness, first-order condition, orientability, and others) is equivalent to the canonical spectral triple associated with a closed Riemannian spin manifold. In this sense, spectral triples provide a complete algebraic encoding of Riemannian manifolds.

In this work, we adopt this theorem as a guiding principle: noncommutative spectral triples are regarded as genuine noncommutative Riemannian manifolds \cite{Connes:1994yd,connes1985non}, whose geometric data are encoded spectrally. As will be shown in later sections, this geometric perspective naturally interfaces with quantum information-theoretic constructions.

\medskip

In summary, spectral triples provide a unifying framework in which topology, metric geometry, and spectral data are encoded algebraically. Noncommutative spectral triples are therefore naturally interpreted as noncommutative Riemannian manifolds, with distance measured by the Connes formula and geometry governed by the Dirac operator.

\subsection{Global Information and Local Transformations}

In this subsection, we clarify how global information and local transformations are distinguished within the framework of spectral triples. Our guiding principle is that global information should be insensitive to local operations, while local transformations are those confined to finite regions with respect to the Connes distance. This perspective will later allow us to identify natural candidates for quantum error correcting code spaces.

\subsubsection{A Discrete Example and Global Degrees of Freedom}

We begin with a simple discrete example to illustrate how global information arises naturally from the spectral properties of the Dirac operator.

Let
\[
    X=\{1,2,\dots,n\}
\]
be a finite set equipped with distances $d_{ij}$ between points. The associated commutative algebra is
\[
    A=C(X)\cong \mathbb{C}^n,
\]
with standard basis elements $e_i(j)=\delta_{ij}$.

We define the Hilbert space
\[
    \mathcal{H}=\bigoplus_{i<j}\mathbb{C}^2
\]
and the Dirac operator
\[
    D=\bigoplus_{i<j}
    \begin{pmatrix}
        0 & d_{ij}^{-1} \\
        d_{ij}^{-1} & 0
    \end{pmatrix}.
\]
For a function $f\in C(X)$, the commutator with its representation satisfies
\[
    [D,\pi(f)] \sim \frac{f_i-f_j}{d_{ij}},
\]
and the Connes distance reproduces the original metric,
\[
    d_D(\omega_i,\omega_j)=d_{ij}.
\]

In this construction, the inverse Dirac operator $D^{-1}$ scales with the distance $d_{ij}$ and therefore controls correlations between points: distant points are weakly correlated, while nearby points are strongly correlated. This behavior reflects the locality encoded by the spectral properties of $D$. However, this description breaks down for the zero modes of $D$.

Indeed, states in $\ker D$ are not suppressed by $D^{-1}$ and correspond to infinitely long-range correlations, independent of the distances $d_{ij}$. Such modes are insensitive to the local metric structure and cannot be distinguished by any operation that probes only finite distances. For this reason, $\ker D$ naturally represents global degrees of freedom of the system.

This observation provides a concrete realization of the intuitive idea underlying quantum error correction: global information is encoded in degrees of freedom that are delocalized over the entire space and cannot be affected by local perturbations. In the spectral triple framework, this intuition is captured precisely by the kernel of the Dirac operator.

\subsubsection{Local Regions and Local Algebras}

In general noncommutative spaces, it is neither necessary nor desirable to regard the entire pure state space $PS(A)$ as the physical or informational space. Instead, we fix a subset
\[
    X \subset PS(A)
\]
and interpret $X$ as the effective space on which locality is defined. Restricting the Connes distance $d_D$ to $X$ equips $(X,d_D)$ with the structure of a metric space.
All notions of locality and local transformations introduced below are defined relative to this metric structure.

Given the metric space $(X,d_D)$, for any $\omega\in X$ and $r>0$ we define the closed ball
\[
    B(\omega,r)=\{\sigma\in X \mid d_D(\sigma,\omega)\le r\}.
\]
Let $\mathcal{I}$ denote the family of subsets of $X$ that can be written as finite unions of such balls,
\[
    Y=\bigcup_{k=1}^m B(\omega_k,r_k).
\]
Elements of $\mathcal{I}$ represent local regions of the space.

For each $Y\in\mathcal{I}$, we define the associated local algebra by
\[
    A_Y
    :=\{\,a\in A_0 \mid \omega(a^*a)=0 \text{ for all } \omega\in X\setminus Y\,\}.
\]
Elements of $A_Y$ are operators localized within the region $Y$. In the commutative case, this definition reduces to functions supported on $Y$.

The family $\{A_Y\}_{Y\in\mathcal{I}}$ satisfies the following properties.

\begin{lemma}
    Let $Y,Z\in\mathcal{I}$. Then:
    \begin{enumerate}
        \item If $Y\subset Z$, then $A_Y \subset A_Z$.
        \item For any $a_0\in A_0$, the set
        \[
            \{\omega\in X \mid \omega(a_0^*a_0)\neq 0\}
        \]
        is compact in $(X,d_D)$. Consequently, there exists $Y\in\mathcal{I}$ such that $a_0\in A_Y$.
    \end{enumerate}
\end{lemma}

\begin{proof}
    (1) If $Y\subset Z$, then $X\setminus Z\subset X\setminus Y$. For any $a\in A_Y$ and any $\omega\in X\setminus Z$, we have $\omega(a^*a)=0$, hence $a\in A_Z$.
    
    (2) For $a_0\in A_0$, the map $\omega\mapsto \omega(a_0^*a_0)$ is a nonnegative continuous function on $X$. Its support is compact with respect to the metric topology induced by $d_D$ and can therefore be covered by finitely many balls, yielding $a_0\in A_Y$ for some $Y\in\mathcal{I}$.
\end{proof}

\subsubsection{Diameter, Support, and Local Transformations}

To quantify the spatial extent of local regions, we define the diameter.

\begin{definition}[Diameter]
    For $Y\in\mathcal{I}$, the diameter of $Y$ is defined by
    \[
        \mathrm{diam}(Y)
        := \sup\{\, d_D(\omega,\sigma) \mid \omega,\sigma \in Y \,\}.
    \]
\end{definition}

Using this notion, we define the support of an operator.

\begin{definition}[Support and diameter of an operator]
    For $E\in A$, let $\mathrm{supp}(E)$ denote the minimal $Y\in\mathcal{I}$ such that $E\in A_Y$. The diameter of $E$ is defined as
    \[
        \mathrm{diam}(E):=\mathrm{diam}(\mathrm{supp}(E)).
    \]
\end{definition}

Operators with finite diameter generate local transformations, acting nontrivially only within bounded regions of $(X,d_D)$. By contrast, degrees of freedom associated with $\ker D$ are insensitive to all such local transformations and therefore encode global information.

This separation between global information and local transformations will be central to the definition of spectral codes in the next section.

\section{Definition and some properties of spectral code}
This section introduces spectral codes, which are the main subject of this paper, and investigates their fundamental properties.
Section~3.1 defines the notion of a spectral code.
Section~3.2 discusses correctable errors for spectral codes.
Section~3.3 shows that a variety of existing error-correcting codes can be reinterpreted as spectral codes.

\subsection{Definition of spectral code}

In this subsection, we introduce the notion of a spectral code, which provides a geometric definition of quantum error correcting codes based on spectral triples. We begin by briefly recalling the standard formulation of quantum error correction, and then explain how it naturally arises from the low-energy structure of a noncommutative Riemannian manifold.

A quantum error correcting code \cite{beny2007generalization,beny2007quantum} is defined as a subspace
\[
    \mathcal{C} = P\mathcal{H}
\]
of a Hilbert space $\mathcal{H}$, where $P$ is an orthogonal projection. The subspace $\mathcal{C}$ encodes logical quantum information. Given a set of error operators $\{E_a\}$ acting on $\mathcal{H}$, the Knill--Laflamme condition \cite{knill1996theory} states that the errors are correctable on $\mathcal{C}$ if and only if
\[
    P E_a^\dagger E_b P = \lambda_{ab} P
\]
for some coefficients $\lambda_{ab}$. This condition implies that, when restricted to the code space, errors act trivially up to an overall scalar and can therefore be reversed.

In this formulation, the code space is specified by a projection, and errors are described by operators belonging to a $C^*$-algebra acting on $\mathcal{H}$.

As discussed in the previous sections, a spectral triple $(A,\mathcal{H},D)$ provides an algebraic encoding of the geometric data of a (possibly noncommutative) Riemannian manifold. The Dirac operator $D$ determines the metric and correlation structure, while its spectrum provides a natural notion of energy scale.

In particular, the low-energy sector of $D$, and especially its zero modes, correspond to global degrees of freedom that are insensitive to the local metric structure. These modes are not suppressed by the inverse Dirac operator and therefore represent infinitely long-range correlations. From the perspective developed in Section~2.3, they encode global information that cannot be affected by local transformations.

This observation motivates defining a quantum code space by projecting onto the low-energy subspace of the Dirac operator. In this way, the intuitive principle of quantum error correction—encoding information into global degrees of freedom protected from local errors—is realized as a geometric construction.

We now give the central definition of this paper.

\begin{definition}[Spectral code]
    Let $(A,\mathcal{H},D)$ be a locally compact spectral triple, in the sense that the Dirac operator induces a metric structure and the local algebras $A_Y$ introduced in Section~2.3 are well-defined.
    
    Let
    \[
        P := \mathbf{1}_{\{0\}}(D)
    \]
    be the spectral projection onto the zero-eigenspace of the Dirac operator. The subspace
    \[
        \mathcal{C} := P\mathcal{H}
    \]
    is referred to as the code space associated with the spectral triple.
    
    We further define the compressed algebra and operator by
    \[
        A_c := P A P, \qquad D_c := P D P = 0,
    \]
    and refer to the triple
    \[
        (A_c,\mathcal{C},D_c)
    \]
    as the \emph{code spectral triple}.
    In this work, an error-correcting code obtained via such a zero-mode projection within the spectral formalism will be called a \emph{spectral code}.
\end{definition}

Although Definition~9 singles out the exact zero-mode sector via the projection $P=\mathbf{1}_{\{0\}}(D)$, it is often natural—both physically and for approximate error correction—to allow a small but finite energy window around zero.
Given a cutoff $\Lambda>0$, we may replace $P$ by the spectral projection onto the low-energy sector
\[
    P_{\Lambda}:=\mathbf{1}_{[-\Lambda,\Lambda]}(D)
    \quad\text{(equivalently, }P_{\Lambda}=\mathbf{1}_{[0,\Lambda]}(|D|)\text{),}
\]
and define the corresponding code space $\mathcal{C}_{\Lambda}:=P_{\Lambda}\mathcal{H}$.
The compressed algebra and operator are then
\[
    A_{c,\Lambda}:=P_{\Lambda}AP_{\Lambda}, \qquad D_{c,\Lambda}:=P_{\Lambda}DP_{\Lambda},
\]
yielding a ``low-energy'' (or band-limited) version of a spectral code, with the zero-mode case recovered at $\Lambda=0$.
This construction extracts the global degrees of freedom of the original noncommutative Riemannian manifold, while discarding local, high-energy modes
(fig.~\ref{fig:spectral}).

\begin{figure}[t]
    \begin{center}
        \includegraphics[width=100mm]{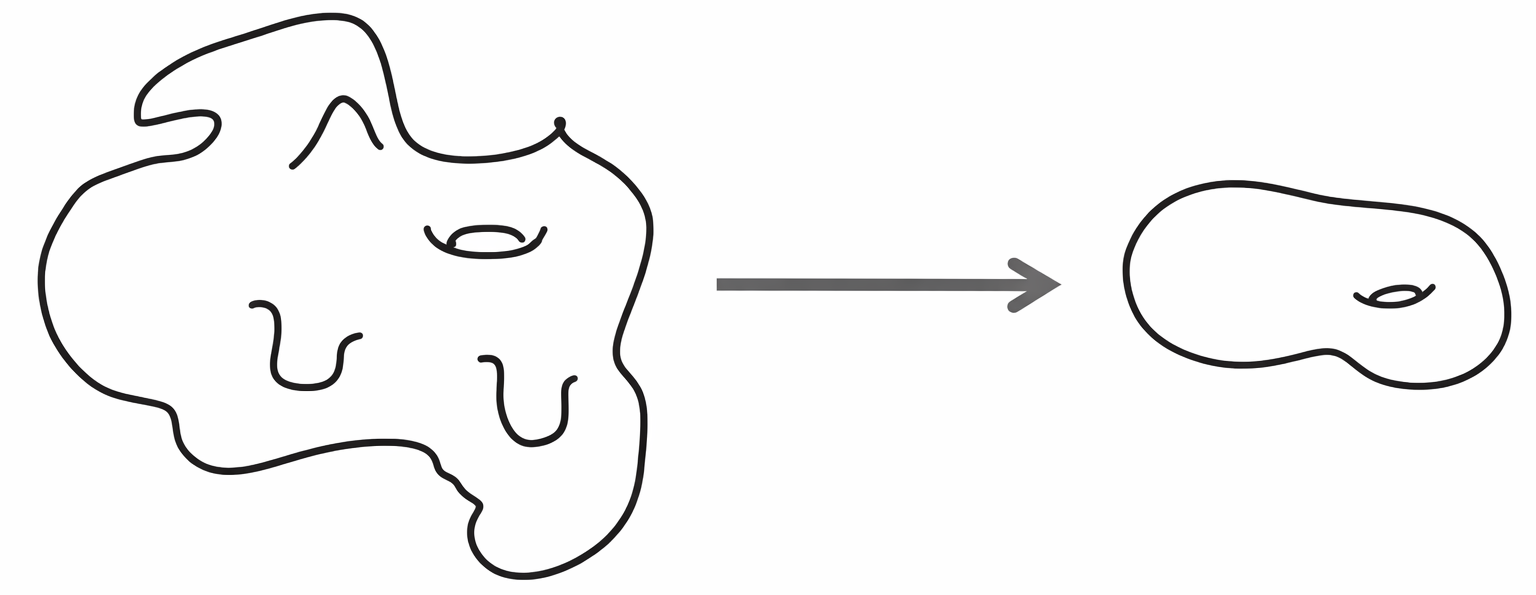}
        \caption{
        Conceptual Figure of spectral code
        }\label{fig:spectral}
    \end{center}
\end{figure}




Finally, the projection
\[
    (A,\mathcal{H},D) \longrightarrow (A_c,\mathcal{C},D_c)
\]
may be interpreted as a low-energy effective description associated with the original noncommutative Riemannian manifold. The spectral code retains only the global modes, which form a protected subspace for encoding quantum information.
In the next section, we prove that this code space remains protected against local errors whose support has bounded diameter.

This perspective will be essential in the subsequent sections, where we demonstrate how known quantum error correcting codes can be recovered within this framework and how controlled geometric deformations lead to improvements in error tolerance.

\subsection{Geometric realization of the Knill--Laflamme condition}

In this subsection, we show how the Knill--Laflamme (KL) conditions arise naturally from the geometric notion of locality introduced in Section~2.3. A key point is that locality is not defined on the entire pure state space $PS(A)$, but rather on a physically relevant subset
\[
    X \subset PS(A),
\]
equipped with the Connes distance $d_D$. All notions of local regions, local algebras, and diameters are understood with respect to this restricted metric space.

Let $P$ be the projection defining a spectral code, typically given by the zero-mode projection $P=\mathbf{1}_{\{0\}}(D)$. We introduce a geometric measure of the robustness of the code against local errors.

\begin{definition}[Code distance]
    The code distance associated with the projection $P$ is defined by
    \[
        d_D(P)
        := \inf \Bigl\{\, \mathrm{diam}(Y) \ \Big|\ Y\in\mathcal{I},\ \exists a\in A_Y \text{ such that } PaP \notin \mathbb{C}P \Bigr\}.
    \]
\end{definition}

Intuitively, $d_D(P)$ represents the minimal spatial scale at which a local operator can act nontrivially on the code space. Operators supported on regions of diameter smaller than $d_D(P)$ act as scalars when restricted to the code space.
This interpretation makes precise the analogy with the distance of a quantum error correcting code.

This intuition is formalized by the following lemma.

\begin{lemma}
    Let $Y\in\mathcal{I}$ satisfy $\mathrm{diam}(Y)< d_D(P)$. Then, for any $a\in A_Y$,
    \[
        PaP \in \mathbb{C}P.
    \]
\end{lemma}

\begin{proof}
    Suppose, for contradiction, that there exists $a\in A_Y$ such that $PaP \notin \mathbb{C}P$. By the definition of the code distance, this implies $d_D(P)\le \mathrm{diam}(Y)$, which contradicts the assumption $\mathrm{diam}(Y)< d_D(P)$. Hence $PaP$ must be proportional to $P$.
\end{proof}

To analyze error correction, one must control the locality of products of error operators. We therefore impose the following natural assumption on the family of local algebras.

\begin{assumption}[Subadditivity of diameter]
    For any $Y,Z\in\mathcal{I}$ and any $a\in A_Y$, $b\in A_Z$, there exists $W\in\mathcal{I}$ such that
    \[
        ab \in A_W,
        \qquad
        \mathrm{diam}(W)\le \mathrm{diam}(Y)+\mathrm{diam}(Z).
    \]
\end{assumption}

This assumption expresses the idea that the support of a product of local operators is contained in the union of their supports, and that the corresponding spatial extent is bounded by the sum of their diameters.

We are now ready to derive the Knill--Laflamme conditions from geometric locality.

\begin{theorem}[Knill--Laflamme condition from locality]
    Let $\{E_\alpha\}\subset A_0$ be a family of error operators satisfying
    \[
        \mathrm{diam}(E_\alpha)\le R
    \]
    for some $R>0$. If
    \[
        2R < d_D(P),
    \]
    then the errors $\{E_\alpha\}$ are correctable on the code space $\mathcal{C}=P\mathcal{H}$, that is,
    \[
        P E_\alpha^\dagger E_\beta P = \lambda_{\alpha\beta} P
    \]
    for some scalars $\lambda_{\alpha\beta}$.
\end{theorem}

\begin{proof}
    By assumption, $\mathrm{diam}(E_\alpha)\le R$ for all $\alpha$. By the subadditivity of the diameter,
    \[
        \mathrm{diam}(E_\alpha^\dagger E_\beta)
        \le \mathrm{diam}(E_\alpha)+\mathrm{diam}(E_\beta)
        \le 2R.
    \]
    Since $2R< d_D(P)$, we have $\mathrm{diam}(E_\alpha^\dagger E_\beta)< d_D(P)$. Applying the previous lemma with $a=E_\alpha^\dagger E_\beta$ yields
    \[
        P E_\alpha^\dagger E_\beta P \in \mathbb{C}P,
    \]
    which implies the existence of scalars $\lambda_{\alpha\beta}$ such that
    \[
        P E_\alpha^\dagger E_\beta P = \lambda_{\alpha\beta} P.
    \]
    This is precisely the Knill--Laflamme condition.
\end{proof}

This theorem shows that the code distance $d_D(P)$ controls the correctability of local errors: any family of errors whose support diameter is smaller than $d_D(P)/2$ is correctable. Thus, the Knill--Laflamme conditions emerge as a direct consequence of the geometric separation between local and global degrees of freedom induced by the Dirac operator.

\subsection{Realization of Arbitrary Codes}

\subsubsection{Formal Reconstruction of Arbitrary Codes}

In this subsection, we show that the framework of spectral codes is sufficiently general to reproduce \emph{any} quantum error correcting code in a formal sense. By ``formal reconstruction,'' we mean that for any given quantum error correcting code, one can construct a spectral triple whose zero-mode sector coincides with the code space. This construction is purely algebraic and does not rely on any pre-existing geometric interpretation of the code. This result establishes spectral codes as a universal framework encompassing all conventional quantum error correcting codes.

\begin{restatable}{theorem}{Formal}
Let $\mathcal{C}\subset\mathcal{H}$ be a quantum error correcting code with projection $P$ onto $\mathcal{C}$, and let $\{E_a\}\subset\mathcal{B}(\mathcal{H})$ be a set of error operators satisfying the Knill--Laflamme condition
\[
    P E_a^\dagger E_b P = \lambda_{ab} P.
\]
Then there exists a spectral triple $(A,\mathcal{H},D)$ such that
\[
    \mathcal{C} = \ker D,
\]
and the associated spectral code defined by the zero-mode projection of $D$ coincides with the given quantum error correcting code.
\end{restatable}

The proof is constructive and proceeds by explicitly defining an auxiliary algebra and a Dirac-type operator whose kernel reproduces the prescribed code space. The proof of this theorem is given in \ref{app:proof}.

This result shows that spectral codes provide a universal formal framework for quantum error correction. However, the construction presented here is highly non-unique and largely artificial: the resulting spectral triple carries little intrinsic geometric structure and is not expected to arise naturally from physical considerations.

In the following subsections, we therefore focus on more structured and physically meaningful spectral triples, and show how well-known families of quantum error correcting codes—such as classical linear codes\cite{macwilliams1977theory}, stabilizer codes\cite{gottesman1997stabilizer}, GKP codes\cite{gottesman2000encoding}, and topological codes\cite{Kitaev:1997wr}—arise naturally within the spectral code framework.

The construction presented above shows that any quantum error correcting code can be formally realized as a spectral code. This universality admits a natural geometric interpretation: error correcting codes arise as projections of noncommutative spaces.

More precisely, a spectral triple $(A,\mathcal{H},D)$ describes the full set of degrees of freedom of a (possibly noncommutative) Riemannian space across all energy scales. The associated spectral code is obtained by a low-energy projection,
\[
    (A,\mathcal{H},D)
    \;\longrightarrow\;
    (A_c,\mathcal{C},D_c),
\]
where $\mathcal{C}=\ker D$ consists solely of global degrees of freedom. All local structures and high-energy modes are discarded in this projection. In this sense, the spectral code captures precisely the long-range, geometry-insensitive sector of the underlying space. In this sense, an error correcting code may be viewed as the low-energy effective theory of a noncommutative space, or equivalently as a projection extracting its global sector.

From this perspective, the code projection $P$ defining the code space is interpreted geometrically as a projection onto the global structure of a noncommutative space. Error correction is therefore not introduced as an ad hoc quantum-information-theoretic mechanism, but rather emerges as a manifestation of a geometric operation within noncommutative geometry.

An important consequence of this viewpoint is that classical and quantum error correcting codes are treated on equal footing. In the commutative case, spectral triples reduce to ordinary topological spaces, and their low-energy projections reproduce classical error correcting codes. In the noncommutative case, the same construction yields quantum error correcting codes. Thus, the distinction between classical and quantum codes corresponds to the distinction between commutative and noncommutative geometries.

Moreover, the present framework naturally extends beyond a unified treatment of classical and quantum error-correcting codes to situations in which classical and quantum structures coexist. In particular, it allows one to consider codes in which classical information is protected by quantum degrees of freedom, as well as codes in which quantum information is protected by classical structures. In this work, we discuss Berezin–Toeplitz quantization as an example of the latter in Section~5, and holography as an example of the former in Section~6 (fig.~\ref{fig:real}).

\begin{figure}[t]
    \begin{center}
        \includegraphics[width=100mm]{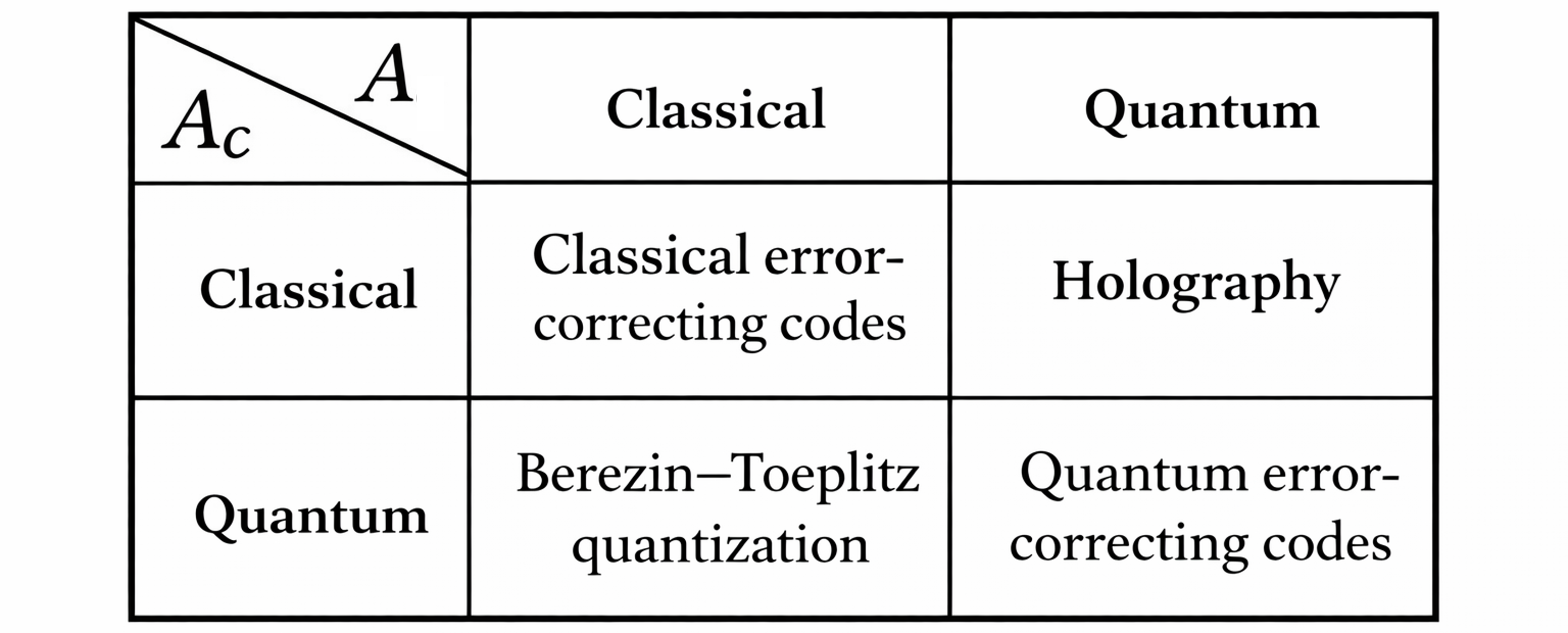}
        \caption{
        Classification of error-correcting structures according to the nature of the algebra and the encoded information.
        Commutative algebras give rise to classical error-correcting codes, while noncommutative algebras yield quantum error-correcting codes. Mixed cases naturally include Berezin–Toeplitz quantization and holography, illustrating that classical and quantum codes arise on equal footing within the spectral framework.
        }\label{fig:real}
    \end{center}
\end{figure}

In summary, arbitrary error correcting codes can be understood as spectral codes arising from projections of noncommutative spaces. This observation provides a unified geometric framework for both classical and quantum error correction. In the following subsections, we restrict attention to spectral triples with additional geometric structure and demonstrate how familiar families of codes naturally emerge from intrinsic geometric constructions rather than purely formal ones.

\subsubsection{Preparations}

In this subsection, we prepare the geometric framework necessary to reconstruct classical linear codes, stabilizer codes, GKP codes, and topological codes within the spectral code formalism. The central idea is to realize a discrete metric structure as a finite spectral triple in a standard way, such that the code distance is recovered as a geometric weight.

\subsubsection*{Discrete Abelian Groups and Weight Functions}

Let $V$ be an abelian group, finite or countable, equipped with a weight function
\[
    \mathrm{wt}: V \to \mathbb{R}_{\ge 0},
\]
satisfying
\[
    \mathrm{wt}(0)=0, \qquad \mathrm{wt}(v)>0 \quad (v\neq 0).
\]
The weight function plays the role of a distance generator and will later be identified with the Connes distance induced by a spectral triple.

We denote by
\[
    C(V) := \{ f: V \to \mathbb{C} \}
\]
the commutative $C^*$-algebra of complex-valued functions on $V$.

\subsubsection*{Group Action and Twisted Crossed Product Algebra}

We define an action of $V$ on $C(V)$ by translations. For each $u\in V$, let
\[
    \alpha_u: C(V) \to C(V), \qquad (\alpha_u f)(x)=f(x-u).
\]

Let $\sigma: V\times V \to U(1)$ be a normalized $2$-cocycle satisfying
\[
    \sigma(u,v)\sigma(u+v,w)=\sigma(v,w)\sigma(u,v+w)
    \quad \text{for all } u,v,w\in V.
\]

The twisted crossed product $C^*$-algebra
\[
    C(V)\rtimes_\sigma V
\]
is generated by $C(V)$ and unitary elements $\{U_u\}_{u\in V}$ subject to the relations
\[
    U_uU_v=\sigma(u,v)U_{u+v},\qquad
    U_u^*=\overline{\sigma(u,-u)}\,U_{-u},\qquad
    U_u f U_u^*=\alpha_u(f).
\]
Any element $a\in C(V)\rtimes_\sigma V$ can be written formally as
\[
    a=\sum_{u\in V} f_u\,U_u, \qquad f_u\in C(V).
\]

This algebra simultaneously encodes position degrees of freedom through $C(V)$ and translation degrees of freedom through the unitaries $U_u$, which will later be interpreted as error operators.

\subsubsection*{Hilbert Space from Pairwise Geometry}

To encode the metric structure, we introduce an auxiliary Hilbert space based on \emph{pairs of points}:
\[
    \mathcal H_m := \bigoplus_{x\neq y\in V} \mathbb C^2,
\]
or equivalently, to avoid duplication,
\[
    \mathcal H_m := \bigoplus_{x<y} \mathbb C^2.
\]
Each summand corresponds to an unordered pair $(x,y)$ and represents the geometric degree of freedom measuring the distance between the two points. This construction coincides with the standard realization of finite spectral triples associated with discrete metric spaces.

\subsubsection*{Dirac Operator}

Using the weight function, we define the Dirac operator $D_m$ on $\mathcal H_m$ by
\[
    D_m
    := \bigoplus_{x\neq y}
    \begin{pmatrix}
        0 & \mathrm{wt}(x-y)^{-1} \\
        \mathrm{wt}(x-y)^{-1} & 0
    \end{pmatrix}.
\]
Thus, pairs of points with small separation correspond to high-energy modes, while distant pairs correspond to low-energy modes. The inverse $D_m^{-1}$ is proportional to the distance and therefore encodes correlation strength.
This is precisely the standard construction of a finite spectral triple reproducing a discrete metric via the Connes distance.

\subsubsection*{Representation and Finite Spectral Triple}

We now define a representation $\pi_m$ of $C(V)\rtimes_\sigma V$ on $\mathcal H_m$. For an element
\[
    a=\sum_{u\in V} f_u\,U_u,
\]
the operator $\pi_m(a)$ acts on each $(x,y)$-component as
\[
    \pi_m(a)\big|_{(x,y)}
    =
    \begin{pmatrix}
        f_0(x) & 0 \\
        0 & f_0(y)
    \end{pmatrix}.
\]
That is, $\pi_m$ acts diagonally by evaluation of the coefficient $f_0$ at the two points.

With this definition, the commutator with the Dirac operator is
\[
    [D_m,\pi_m(a)]\big|_{(x,y)}
    =
    \mathrm{wt}(x-y)^{-1}
    \begin{pmatrix}
        0 & f_0(x)-f_0(y) \\
        f_0(y)-f_0(x) & 0
    \end{pmatrix},
\]
which is bounded. Hence
\[
    \bigl(C(V)\rtimes_\sigma V,\ \mathcal H_m,\ D_m\bigr)
\]
defines a finite (or locally finite) spectral triple, and the induced Connes distance satisfies
\[
    d_{D_m}(\omega_x,\omega_y)=\mathrm{wt}(x-y)
\]
for pure states $\omega_x,\omega_y$.

\subsubsection*{Physical Hilbert Space and Full Spectral Triple}

We now introduce the physical Hilbert space
\[
    \mathcal H_{\mathrm{phys}} := \ell^2(V),
\]
with canonical basis $\{\ket{x}\}_{x\in V}$. The regular representation $\pi_{\mathrm{reg}}$ is defined by
\[
    \pi_{\mathrm{reg}}(f)\ket{x}=f(x)\ket{x},\qquad
    \pi_{\mathrm{reg}}(U_u)\ket{x}=\sigma(u,x)\ket{x+u}.
\]

The full Hilbert space is
\[
    \mathcal H := \mathcal H_m \oplus \mathcal H_{\mathrm{phys}},
\]
and the Dirac operator is
\[
    D := D_m \oplus D_c,
\]
where $D_c$ acts on $\mathcal H_{\mathrm{phys}}$.
Since $D_m$ has no zero modes by construction,
we have $\ker D_m=\{0\}$.
Consequently, the kernel of the full Dirac operator is
\[
    \ker D
    = \ker D_m \oplus \ker D_c
    = \{0\}\oplus\ker D_c
    \simeq \ker D_c
    = \mathcal C,
\]
where $\mathcal C$ denotes the code space, as defined below.

\subsubsection*{Choice of Pure States and Recovery of the Metric}

We select a subset of pure states
\[
    X := \{\omega_x \mid \omega_x(a)=f_0(x) \text{ for } a=\sum_u f_uU_u\}
    \subset PS\bigl(C(V)\rtimes_\sigma V\bigr),
\]
interpreted as the position space. Restricting the Connes distance to $X$, we recover the weight as distance.

\begin{restatable}{lemma}{disred}
    \label{lem:distance-reduction}
    For all $\omega_x,\omega_y\in X$, one has
    \[
        d_D(\omega_x,\omega_y)=d_{D_m}(\omega_x,\omega_y)=\mathrm{wt}(x-y).
    \]
\end{restatable}

The proof of this lemma is given in \ref{app:proof}.

\subsubsection*{Local Algebras and Code Distance}

Finally, the local algebras introduced in Section~2.3 admit the following concrete description in the discrete crossed-product setting.

\begin{restatable}{lemma}{ay}
    \label{lem:AY-support}
    Let $A=C(V)\rtimes_\sigma V$ and write $a\in A$ as a finite sum
    \[
        a=\sum_{u\in V} f_u U_u, \qquad f_u\in C(V).
    \]
    Fix a subset $X\subset PS(A)$ of ``position states'' $\{\omega_x\}_{x\in V}$, and let $Y\subset X$.
    For such a subset $Y$, we define
    \[
        A_Y:=\{a\in A_0\mid \forall\,\omega\in X\setminus Y,\ \omega(a^*a)=0\}.
    \]
    Then one has the following characterization:
    \[
        A_Y=\Bigl\{\sum_{u\in V} f_uU_u\in A \ \Bigm|\ \forall u\in V,\ supp(f_u)\subset Y+u\Bigr\}, 
    \]
    where
    \[
        Y+u:=\{y+u\mid y\in Y\}\subset V,
    \]
    and we identify the subset $Y\subset X$ with the corresponding subset of $V$ via the canonical correspondence $Y=\{\omega_y\mid y\in V\}$.
\end{restatable}

The proof of this lemma is given in \ref{app:proof}.

This characterization shows that locality is encoded entirely in the supports of the coefficient functions. It therefore provides a concrete and computable notion of locality in the crossed-product algebra. It provides the technical foundation for computing code distances and for reconstructing concrete error correcting codes in the subsequent subsections.

We also recall the definition
\[
    W:=\{u\in V\mid P\,\pi_{\mathrm{reg}}(U_u)\,P\notin\mathbb C P\},
\]
which identifies those translation operators that act nontrivially on the code space.

\begin{restatable}{theorem}{cc}
    \label{thm:classical-distance}
    The code distance of the spectral code satisfies
    \[
        d_D(P)=\min_{u\in W}\mathrm{wt}(u).
    \]
\end{restatable}

The proof of this theorem is given in \ref{app:proof}.

\subsubsection*{Summary}

We summarize the crossed-product construction and the resulting distance formula for the code.

Let $V$ be a finite abelian group (or a finite subset endowed with an additive structure), and let
\[
    A:=C(V)\rtimes_{\sigma} V
\]
be the twisted crossed-product $C^*$-algebra determined by a $2$-cocycle $\sigma$ (equivalently, by a phase function $\omega$ in the regular representation).
We consider the spectral triple on the Hilbert space
\[
    \mathcal{H}:=\mathcal{H}_m\oplus \mathcal{H}_{\mathrm{phys}},
    \qquad
    D:=D_m\oplus D_c,
    \qquad
    \pi:=\pi_m\oplus \pi_{\mathrm{reg}},
\]
where the ``metric part'' $(C(V)\rtimes_\sigma V,\mathcal{H}_m,D_m)$ encodes a weight function $\mathrm{wt}$ through the choice of Dirac operator, and the ``code part'' $(C(V)\rtimes_\sigma V,\mathcal{H}_{\mathrm{phys}},D_c)$ satisfies
\[
    \mathcal{C}:=\ker D_c,\qquad
    P_c:\mathcal{H}_{\mathrm{phys}}\to \mathcal{C}
    \ \ \text{(orthogonal projection)}.
\]

We define the set of nontrivial translations by
\[
    W:=\Bigl\{u\in V\ \Bigm|\ P_c\,\pi_{\mathrm{reg}}(U_u)\,P_c \notin \mathbb{C}P_c\Bigr\}.
\]
Then the code distance associated to $P_c$ (as defined via geometric localization by diameter) is given by
\[
    d_D(P_c)=\min_{u\in W}\mathrm{wt}(u).
\]

In particular, by choosing $(V,D_c,\mathrm{wt})$ appropriately —and hence fixing the cocycle data $(\sigma,\omega)$ through the regular representation— one can reproduce the desired distance property in this crossed-product model.

\subsubsection{Classical Linear Codes}

In this subsection, we show that the spectral code framework naturally reproduces classical linear codes. In particular, we demonstrate that the code distance emerges geometrically as the Hamming weight.

Let
\[
    V=\mathbb{F}_2^n
\]
be a vector space over the finite field $\mathbb{F}_2$.
The physical Hilbert space is taken to be
\[
    \mathcal{H}_{\mathrm{phys}}=\ell^2(\mathbb{F}_2^n),
\]
with canonical orthonormal basis $\{\ket{x}\}_{x\in\mathbb{F}_2^n}$.

Let $C\subset\mathbb{F}_2^n$ be a classical linear code.
We define the code subspace as
\[
    \mathcal{C}=\ell^2(C)=\mathrm{span}\{\ket{x}\mid x\in C\},
\]
and denote by
\[
    P:=\sum_{x\in C}\ket{x}\bra{x}
\]
the orthogonal projection onto $\mathcal{C}$.

We define the ``code part'' of the Dirac operator by
\[
    D_c\ket{x}=
    \begin{cases}
        0\cdot\ket{x}, & x\in C,\\
        1\cdot\ket{x}, & x\notin C.
    \end{cases}
\]
Then
\[
    \ker D_c=\ell^2(C)=\mathcal{C},
\]
so the code space is realized as the zero-mode subspace of $D_c$.

On the group $V=\mathbb{F}_2^n$, we choose the Hamming weight
\[
    \mathrm{wt}(u):=\#\{\,i\mid u_i\neq 0\,\}
\]
as the weight function.
As shown in Section~3.3.2, this choice is geometrically implemented by the Dirac operator $D_m$ through the Connes distance:
\[
    d_D(\omega_x,\omega_y)=\mathrm{wt}(x-y).
\]

Thus, the full spectral triple
\[
    \bigl(C(\mathbb{F}_2^n)\rtimes \mathbb{F}_2^n,\ 
    \mathcal{H}_m\oplus\mathcal{H}_{\mathrm{phys}},\ 
    D_m\oplus D_c\bigr)
\]
encodes the Hamming geometry of $\mathbb{F}_2^n$.

Recall the definition
\[
    W:=\{\,u\in V\mid P\,\pi_{\mathrm{reg}}(U_u)\,P\notin\mathbb{C}P\,\},
\]
where $\pi_{\mathrm{reg}}(U_u)\ket{x}=\ket{x+u}$.

If $u\notin C$, then $\pi_{\mathrm{reg}}(U_u)$ maps $\ell^2(C)$ into its orthogonal
complement, and hence
\[
    P\,\pi_{\mathrm{reg}}(U_u)\,P=0\in\mathbb{C}P.
\]

On the other hand, if $u\in C\setminus\{0\}$, then $\pi_{\mathrm{reg}}(U_u)$ preserves $\ell^2(C)$ but acts by a nontrivial permutation of the basis vectors $\{\ket{x}\mid x\in C\}$. Therefore, it is not a scalar multiple of the identity on $\mathcal{C}$, and
\[
    P\,\pi_{\mathrm{reg}}(U_u)\,P\notin\mathbb{C}P.
\]

Consequently, we obtain
\[
    W=C\setminus\{0\}.
\]

By the general distance formula established in Section~3.3.2,
\[
    d_D(P)=\min_{u\in W}\mathrm{wt}(u).
\]
Substituting $W=C\setminus\{0\}$ yields
\[
    d_D(P)=\min_{0\neq c\in C}\mathrm{wt}(c),
\]
which is exactly the minimum distance of the classical linear code $C$.

We conclude that classical linear codes are naturally realized as spectral codes.
Their minimum distance arises as the Connes distance induced by the Dirac operator and coincides with the Hamming distance. This provides a geometric reinterpretation of classical coding theory within the framework of noncommutative geometry.

\subsubsection{Stabilizer Codes}

In this subsection, we show that stabilizer codes arise naturally within the spectral code framework. In contrast to the classical case, the underlying algebra is essentially noncommutative, and the quantum nature of stabilizer codes is encoded geometrically through a twisted crossed-product construction.

We take as the phase space
\[
    V=\mathbb{F}_2^{2n}.
\]
An element $v\in V$ is written as
\[
    v=(p\mid q), \qquad p,q\in\mathbb{F}_2^n.
\]
We equip $V$ with the standard symplectic form
\[
    \langle (p\mid q),(p'\mid q')\rangle
    := p\cdot q' + q\cdot p' \pmod{2}.
\]

As a weight function, we choose the Pauli (Hamming) weight
\[
    \mathrm{wt}(p\mid q)
    := \#\{\,i \mid (p_i,q_i)\neq (0,0)\,\},
\]
which counts the number of qubits on which the corresponding Pauli operator acts nontrivially.

For each $v\in V$, we introduce a Weyl operator $U_v$ satisfying
\[
    U_v U_{v'} = \sigma(v,v')\,U_{v+v'},
    \qquad
    \sigma(v,v'):=(-1)^{\langle v,v'\rangle}.
\]
Equivalently,
\[
    U_v U_{v'} = (-1)^{\langle v,v'\rangle} U_{v'} U_v,
\]
which reproduces the canonical commutation relations of Pauli operators.
The associated twisted crossed-product algebra
\[
    C(V)\rtimes_\sigma V
\]
thus provides an algebraic realization of the Pauli group.

Let $L\subset V$ be an isotropic subspace, i.e.,
\[
    \langle \ell,\ell'\rangle = 0
    \quad
    (\forall\,\ell,\ell'\in L).
\]
The stabilizer group is generated by $\{U_\ell\mid \ell\in L\}$.
If $L=\mathrm{span}\{\ell_1,\dots,\ell_m\}$, we define the code part of the Dirac
operator by
\[
    D_c := \sum_{i=1}^m (1-U_{\ell_i}).
\]

Then the kernel of $D_c$ is given by
\[
    \ker D_c
    =
    \{\,\psi\in\mathcal{H}_{\mathrm{phys}}
    \mid
    U_{\ell_i}\psi=\psi\ \text{for all } i\,\},
\]
which coincides with the stabilizer code space. Hence, stabilizer codes are realized as zero-mode subspaces of the Dirac operator.

Combining the metric part $D_m$ introduced in Section~3.3.2 with the code part $D_c$, we obtain the spectral triple
\[
    \bigl(
    C(\mathbb{F}_2^{2n})\rtimes_\sigma \mathbb{F}_2^{2n},\ 
    \mathcal{H}_m\oplus\mathcal{H}_{\mathrm{phys}},\ 
    D_m\oplus D_c
    \bigr).
\]
The Connes distance induced by $D_m$ satisfies
\[
    d_D(\omega_v,\omega_{v'})=\mathrm{wt}(v-v'),
\]
so the geometric distance coincides with the Pauli weight.

Recall the definition
\[
    W:=\{\,v\in V\mid P\,\pi_{\mathrm{reg}}(U_v)\,P\notin\mathbb{C}P\,\},
\]
where $P$ is the projection onto $\ker D_c$.

If $v\notin L^\perp$, then there exists $\ell\in L$ such that $\langle \ell,v\rangle=1$. In this case,
\[
 U_\ell U_v = - U_v U_\ell.
\]
For any $\psi\in\ker D_c$, we have $U_\ell\psi=\psi$, and therefore
\[
 U_\ell(U_v\psi)=-U_v\psi,
\]
which implies $U_v\psi\notin\ker D_c$. Hence,
\[
 P\,U_v\,P = 0 \in \mathbb{C}P.
\]

If $v\in L^\perp$, then $\langle \ell,v\rangle=0$ for all $\ell\in L$, so $U_v$ commutes with all stabilizers and preserves $\ker D_c$. Moreover, if $v\in L^\perp\setminus L$, then $U_v$ acts as a nontrivial unitary on the code space, and thus
\[
    P\,U_v\,P \notin \mathbb{C}P.
\]

Consequently, we obtain
\[
    W = L^\perp \setminus L.
\]

By the general distance formula established earlier,
\[
    d_D(P)=\min_{v\in W}\mathrm{wt}(v),
\]
we conclude that
\[
    d_D(P)=\min_{v\in L^\perp\setminus L}\mathrm{wt}(v),
\]
which coincides with the standard definition of the distance of a stabilizer
code.

We conclude that stabilizer codes are naturally realized as spectral codes.
Their distance arises as a geometric quantity determined by the Dirac operator,
and the noncommutative structure of the Pauli algebra is encoded by the cocycle
$\sigma$ of the twisted crossed-product algebra. This shows that the spectral code framework captures both the algebraic and geometric content of stabilizer quantum error correcting codes.

In \cite{marcolli2012codes}, stabilizer codes are constructed explicitly in terms of vector bundles over tensor products of rational noncommutative tori, where the code space arises as a joint eigenspace of a commutative subalgebra, equivalently described as a subbundle on which the commutative algebra acts by scalars.
By contrast, our approach is formulated in the more general setting of spectral triples and does not rely on the specific structure of noncommutative tori.
The two constructions therefore provide complementary, but conceptually distinct, noncommutative-geometric realizations of stabilizer codes.

\subsubsection{Discrete Lattice GKP Codes}

In this subsection, we show that discrete lattice versions of GKP codes are naturally reproduced within the spectral code framework.
This example illustrates how continuous-variable quantum error correcting codes admit a discrete counterpart that fits seamlessly into the crossed-product and spectral triple construction.

We take the phase space to be
\[
    V=\mathbb{Z}^2.
\]
Elements are written as $(m,n)\in\mathbb{Z}^2$. We define a symplectic pairing modulo $2$ by
\[
    \langle (m,n),(m',n')\rangle
    := mn' + nm' \pmod{2}.
\]

As a weight function, we choose
\[
    \mathrm{wt}(m,n):=|m|+|n|,
\]
which measures the Manhattan distance on the lattice and plays the role of the elementary displacement error weight.

We define a $2$-cocycle by
\[
    \sigma\bigl((m,n),(m',n')\bigr)
    := (-1)^{\langle (m,n),(m',n')\rangle}.
\]
The associated Weyl operators $\{U_{(m,n)}\}_{(m,n)\in\mathbb{Z}^2}$ satisfy
\[
    U_{(m,n)}U_{(m',n')}
    =
    \sigma\bigl((m,n),(m',n')\bigr)\,
    U_{(m+m',\,n+n')},
\]
which reproduces the discrete canonical commutation relations.

The resulting twisted crossed-product algebra
\[
    C(\mathbb{Z}^2)\rtimes_\sigma \mathbb{Z}^2
\]
thus provides the algebraic setting for the discrete GKP construction.

We choose the stabilizer lattice
\[
    L = 2\mathbb{Z}^2
    = \{(2m,2n)\mid m,n\in\mathbb{Z}\}
    = \mathrm{span}\{(0,2),(2,0)\}.
\]
In this case, one has
\[
    L^\perp = \mathbb{Z}^2,
\]
with respect to the symplectic pairing.

Let $\mathcal{H}_{\mathrm{phys}}=\ell^2(\mathbb{Z}^2)$ be the physical Hilbert space.
We define the code part of the Dirac operator by
\[
    D_c := (1-U_{(0,2)}) + (1-U_{(2,0)}).
\]
Then
\[
    \ker D_c
    =
    \{\psi\in\mathcal{H}_{\mathrm{phys}}
    \mid
    U_{(0,2)}\psi=\psi,\;
    U_{(2,0)}\psi=\psi
    \},
\]
which is precisely the stabilizer condition defining the discrete lattice GKP
code space.

As in the stabilizer case, we define
\[
    W:=\{\,w\in V\mid P\,\pi_{\mathrm{reg}}(U_w)\,P\notin\mathbb{C}P\,\},
\]
where $P$ is the projection onto $\ker D_c$.
Since
\[
    W=L^\perp\setminus L
    =\mathbb{Z}^2\setminus(2\mathbb{Z}^2),
\]
the code distance is given by
\[
    d_D(P)=\min_{w\in W}\mathrm{wt}(w).
\]
The minimal nonzero weight in $\mathbb{Z}^2\setminus(2\mathbb{Z}^2)$ is $1$, and
therefore
\[
    d_D(P)=1.
\]

We conclude that the discrete lattice GKP code is faithfully reproduced as a spectral code. The stabilizer lattice $L=2\mathbb{Z}^2$ is encoded in the kernel of the Dirac operator, while the error distance is captured geometrically by the weight function through the Connes distance.
This demonstrates that discrete GKP codes\cite{gottesman2000encoding} arise naturally within the spectral code framework, providing a bridge between stabilizer codes and continuous-variable GKP constructions in a unified geometric setting.

\subsubsection{Topological Codes (Quantum Double Model)}

In this subsection, we show that topological quantum error-correcting codes of the Kitaev quantum double type are naturally realized as spectral codes.
We emphasize the equivalence between the algebraic formulation in terms of a crossed-product algebra and the geometric formulation in terms of ribbon operators. Based on this equivalence, we construct the code space via a Dirac operator and show that the code distance is reproduced as the minimal geometric length of non-contractible ribbon operators.

\subsubsection*{(i) Algebra as a Crossed Product}

Let $\Gamma=(V,E,F)$ be a finite cell decomposition of a closed oriented surface, and let $G$ be a finite group.
We assign a group variable $h_e\in G$ to each edge $e\in E$, and define the configuration space
\[
    \mathcal X := G^{E}.
\]
The physical Hilbert space is
\[
    \mathcal H_{\mathrm{phys}} := \ell^2(G^{E}),
\]
with canonical basis $\{\ket{h}\}_{h\in G^{E}}$.

The group of vertex gauge transformations is
\[
    \mathcal G := G^{V}.
\]
For $g=(g_v)_{v\in V}\in G^{V}$ and $h=(h_e)_{e\in E}\in G^{E}$, we define
\[
    (g\cdot h)_e := g_{s(e)}\,h_e\,g_{t(e)}^{-1},
\]
where $s(e)$ and $t(e)$ denote the source and target vertices of the oriented edge $e$.
This induces an action $\alpha$ of $G^{V}$ on $C(G^{E})$ by
\[
    (\alpha_g f)(h) := f(g^{-1}\cdot h).
\]

We define the total algebra of the quantum double model as the crossed product
\[
    A := C(G^{E}) \rtimes_{\alpha} G^{V}.
\]
It acts naturally on $\mathcal H_{\mathrm{phys}}$ via multiplication by functions in $C(G^{E})$ and unitary operators implementing gauge transformations.
This crossed-product algebra provides a $C^*$-algebraic formulation of the quantum double model.

\subsubsection*{(ii) Local Generators and Ribbon Operators}

The crossed-product algebra $A$ can be completely generated by local operators supported on individual edges.
For each edge $e\in E$ and $g,k\in G$, we define operators on
$\mathcal H_{\mathrm{phys}}$ by
\[
    T_e^k \ket{h} := \delta_{h_e,k}\ket{h}, \qquad
    L_e^g \ket{h} := \ket{(g h_e, h_{e'\neq e})}, \qquad
    R_e^g \ket{h} := \ket{(h_e g^{-1}, h_{e'\neq e})}.
\]
Here $T_e^k$ is a magnetic-type projection measuring the value of the edge variable, while $L_e^g$ and $R_e^g$ are electric-type shift operators implementing left and right group actions on the edge variable.

Any multiplication operator in $C(G^{E})$ can be written as a product and linear combination of the operators $T_e^k$.
Moreover, a vertex gauge transformation is implemented as a product of $L_e^{g_{s(e)}}$ and $R_e^{g_{t(e)}}$ over incident edges.
Therefore,
\[
    A = \mathrm{Alg}^\ast\bigl(\{T_e^k,\,L_e^g,\,R_e^g \mid e\in E,\ g,k\in G\}\bigr)
\]
in the natural representation on $\mathcal H_{\mathrm{phys}}$.

We now introduce ribbon operators.
A \emph{ribbon} $\rho$ is defined as a finite sequence of minimal two-dimensional cells (triangles) formed from the lattice $\Gamma$ and its dual lattice $\Gamma^\ast$.
Each triangle consists of a lattice edge together with an adjacent dual edge.
By connecting such triangles consecutively, one obtains a narrow ribbon-shaped two-dimensional region connecting two sites (vertices or faces).
Although ribbons are often depicted as one-dimensional paths, they are intrinsically two-dimensional objects.

There are two types of triangles:
\begin{itemize}
    \item \emph{Direct triangles}, corresponding to steps along a lattice edge and associated with electric-type operations.
    \item \emph{Dual triangles}, corresponding to steps along a dual edge crossing a lattice edge and associated with magnetic-type operations.
\end{itemize}

To each triangle $\tau$ supported on an edge $e=e(\tau)$, we associate a local
triangle operator:
\[
    F_\tau^g :=
    \begin{cases}
    L_e^g, & \text{for a direct triangle oriented along } e,\\
    R_e^g, & \text{for a direct triangle oriented opposite to } e,
    \end{cases}
    \qquad
    F_{\tau^\ast}^k := T_e^k
\]
for a dual triangle $\tau^\ast$ crossing $e$.
(Orientation-dependent inverses may appear depending on conventions, but this does not affect the arguments below.)

Given a ribbon $\rho=\tau_1\cdots\tau_n$ and labels $(g,h)\in G\times G$, the ribbon operator $F_\rho^{g,h}$ is defined as an ordered product of the corresponding triangle operators, with an internal contraction over intermediate labels.
Conceptually, it takes the form
\[
    F_\rho^{g,h}
    =
    \sum_{x_1,\dots,x_{n-1}\in G}
    F_{\tau_1}^{(g,x_1)} F_{\tau_2}^{(x_1,x_2)} \cdots
    F_{\tau_n}^{(x_{n-1},h)}.
\]
The essential point for our purposes is that each factor is built from the local generators $T,L,R$.

We define the \emph{ribbon algebra} by
\[
    \mathcal R :=
    \mathrm{Alg}^\ast\bigl(\{F_\rho^{g,h} \mid \rho \text{ an open ribbon},\ (g,h)\in G^2\}\bigr)
    \subset \mathcal B(\mathcal H_{\mathrm{phys}}).
\]

\subsubsection*{(iii) Equivalence of the Ribbon Algebra and the Crossed-Product Algebra}

\begin{proposition}[Equivalence of ribbon and crossed-product algebras]
In the natural representation on $\mathcal H_{\mathrm{phys}}$, one has
\[
    \mathcal R = A.
\]
\end{proposition}

\begin{proof}
First, $\mathcal R\subset A$.
Each triangle operator is, by definition, one of $T_e^k$, $L_e^g$, or $R_e^g$.
Hence any ribbon operator $F_\rho^{g,h}$ is a finite sum of products of these operators and therefore belongs to $A$.

Conversely, consider a minimal ribbon consisting of a single triangle.
Direct triangles realize the generators $L_e^g$ and $R_e^g$, while dual triangles realize the generators $T_e^k$.
Thus
\[
    \{T_e^k, L_e^g, R_e^g\} \subset \mathcal R.
\]
Since the crossed-product algebra $A$ is generated by these operators, we obtain $A\subset\mathcal R$.
Combining both inclusions yields $\mathcal R=A$.
\end{proof}

\subsubsection*{(iv) Stabilizers and the Code Dirac Operator}

Local elements of the ribbon algebra allow one to construct the stabilizer operators of the quantum double model.
For each vertex $v\in V$, we define the star operator $A_v$ as a product (or group average) of local electric ribbon operators surrounding $v$, enforcing gauge invariance.
Similarly, for each face $f\in F$, we define the plaquette operator $B_f$ as a product of local magnetic ribbon operators around the boundary of $f$, enforcing trivial flux. These operators are commuting projections.

We define the code Dirac operator by
\[
    D_{\mathrm{code}}
    := \sum_{v\in V}(I-A_v) + \sum_{f\in F}(I-B_f).
\]
Its kernel is given by
\[
    \ker D_{\mathrm{code}}
    =
    \{\psi\in\mathcal H_{\mathrm{phys}}
    \mid A_v\psi=\psi,\ B_f\psi=\psi\ \forall v,f\},
\]
which coincides with the ground-state subspace of the quantum double Hamiltonian.
We denote this code space by $\mathcal C$ and let $P$ be the orthogonal projection
onto $\mathcal C$.

\subsubsection*{(v) Weight and Distance}

To relate ribbon operators to geometric distance, we define a weight function.
For a configuration $h\in G^{E}$, let
\[
    \mathrm{wt}(h)
    := \#\{\,e\in E \mid h_e \neq e_G\,\}.
\]
An operator $O$ is said to be supported on a subset $E'\subset E$ if it modifies only the variables $h_e$ with $e\in E'$.
We define
\[
    \mathrm{wt}(O)
    := \min\{\,|E'| \mid O \text{ is supported on } E'\,\}.
\]

A ribbon operator $F_\rho$ acts nontrivially precisely on the edges (and corresponding dual edges) involved in the ribbon.
Therefore, its weight coincides with the geometric length of the ribbon.

\begin{theorem}[Distance of the topological code]
The code distance is given by
\[
    d_D(P)
    =
    \min\{\mathrm{wt}(F_\rho)\mid
    \rho\ \text{is a non-contractible closed ribbon}\}.
\]
\end{theorem}

\begin{proof}
We consider closed ribbon operators.

\emph{(Contractible case).}
If $\rho$ is contractible, then the corresponding operator $F_\rho$ can be written as a product of vertex and plaquette operators.
Hence, on the code space,
\[
    P F_\rho P = \lambda_\rho P
\]
for some scalar $\lambda_\rho$, and $F_\rho$ acts trivially on $\mathcal C$.

\emph{(Non-contractible case).}
If $\rho$ is non-contractible, then $F_\rho$ cannot be generated by local stabilizers.
It acts as a nontrivial logical operator on $\mathcal C$, and therefore
\[
    P F_\rho P \notin \mathbb{C}P.
\]

By definition, the distance of the code is the minimal weight of an operator acting nontrivially on the code space.
Since only non-contractible closed ribbons act nontrivially, and their weights are measured by $\mathrm{wt}(F_\rho)$, the stated formula follows.
\end{proof}

\subsubsection*{Conclusion}

We conclude that topological quantum codes of the quantum double type are naturally realized as spectral codes.
The crossed-product algebra $C(G^{E})\rtimes G^{V}$ is equivalent to the ribbon algebra generated by ribbon operators.
The code space arises as the kernel of the Dirac operator $D_{\mathrm{code}}$, and the code distance is determined geometrically as the minimal length of non-contractible ribbon operators.

\section{Code improvement via inner fluctuations}

This section discusses the possibility of improving the error-correction threshold via internal fluctuations in noncommutative spaces.
Section~4.1 defines a decoder.
Section~4.2 introduces the notion of the threshold.
Section~4.3 shows that, under this definition, the threshold can be described in terms of leakage out of the code space, and that this leakage is controlled by the spectral gap of the Dirac operator in the underlying spectral triple.
Section~4.4 discusses how internal fluctuations can modify the spectral gap and thereby lead to an improvement of the threshold.

\subsection{Definition of the Decoder}

In this subsection, we define a decoder suitable for analyzing threshold improvement induced by inner perturbations.
The construction combines a small-noise expansion of error channels, the Petz recovery map, and a conditional expectation that removes local (high-energy) degrees of freedom while preserving global information.

We consider a code space $\mathcal{C}=\ker D\subset\mathcal{H}$ with orthogonal projection $P$.
Noise acting on the system is modeled by a completely positive
trace-preserving map
\[
    \mathcal{E}_\theta(\rho)
    =
    \sum_i E_i(\theta)\,\rho\,E_i(\theta)^\ast,
\]
where $\theta\ge 0$ is a noise strength parameter.
We assume that, in the small-noise limit $\theta\to 0$, the Kraus operators admit an expansion of the form
\[
    E_0(\theta)=\sqrt{1-\theta}\,I+O(\theta),
    \qquad
    E_i(\theta)=\sqrt{\theta}\,F_i+O(\theta^{3/2})
    \quad (i\ge 1),
\]
where $\{F_i\}$ represent elementary error operators.
In the geometric setting of spectral codes, these operators are typically local with respect to the metric induced by the Dirac operator.

Let $d=\dim\mathcal{C}$ and define the maximally mixed state on the code space by
\[
    \sigma:=\frac{1}{d}P.
\]
Given the channel $\mathcal{E}_\theta$, we introduce the Petz recovery map \cite{Petz:1986tvy} associated with $\sigma$,
\[
    \mathcal{R}_{\sigma,\mathcal{E}_\theta}(X)
    :=
    \sigma^{1/2}\,
    \mathcal{E}_\theta^\ast\!\left(
    \mathcal{E}_\theta(\sigma)^{-1/2}\,
    X\,
    \mathcal{E}_\theta(\sigma)^{-1/2}
    \right)\,
    \sigma^{1/2},
\]
where $\mathcal{E}_\theta^\ast$ denotes the dual of $\mathcal{E}_\theta$.
This recovery map is known to be optimal with respect to the reference state $\sigma$ and plays a central role in approximate quantum error correction \cite{hiai2011quantum,barnum2002reversing,fawzi2015quantum,berta2014monotonicity}.

To model the removal of local excitations, we assume that the Hilbert space admits a factorization
\[
    \mathcal{H}
    =
    \mathcal{H}_{\mathrm{low}}\otimes\mathcal{H}_{\mathrm{high}},
\]
where $\mathcal{H}_{\mathrm{low}}$ corresponds to low-energy (global) degrees of freedom and $\mathcal{H}_{\mathrm{high}}$ to high-energy (local) degrees of freedom.
To eliminate local excitations, we introduce a conditional expectation\cite{TAKESAKI1972306}
\[
    \mathbb{E}:\mathcal{B}(\mathcal{H})\to \mathcal{B}(\mathcal{H}_{\mathrm{low}})
\]
defined by
\[
    \mathbb{E}(X)
    :=
    \frac{1}{\dim\mathcal{H}_{\mathrm{high}}}\,
    \mathrm{Tr}_{\mathrm{high}}(X).
\]
This operation traces out high-energy degrees of freedom and retains only the global components relevant for logical information.

With these ingredients, we define the decoder associated with the noise channel $\mathcal{E}_\theta$ by
\[
    \mathcal{N}_\theta
    :=
    \mathbb{E}\circ
    \mathcal{R}_{\sigma,\mathcal{E}_\theta}\circ
    \mathcal{E}_\theta.
\]
The decoding procedure thus consists of applying the noise channel,
performing Petz recovery, and finally discarding local excitations via the conditional expectation.
As a result, local errors are suppressed, while logical information encoded in low-energy global degrees of freedom is preserved.

This formulation makes explicit the interpretation of decoding as a coarse-grained recovery process.
In the subsequent sections, we analyze how internal perturbations of the Dirac operator modify the structure of local excitations and can lead to improvements in the decoding threshold.

In the present work, we employ analytically derived recovery maps based on conditional expectations and the Petz recovery map. On the other hand, as an alternative approach to designing recovery maps from the perspective of fidelity maximization, Yamamoto et al. proposed a construction of recovery maps based on semidefinite programming, showing the connection between theoretical optimality conditions and numerical design methods for error correction \cite{yamamoto2005suboptimal}.

\subsection{Definition of the Threshold}

In this subsection, we define the error-correction threshold using the decoder introduced in the previous section.
The performance of decoding is quantified by the entanglement fidelity, and the threshold is characterized through the asymptotic behavior of an induced iteration map on the effective error rate.

Let $\mathcal{C}=\ker D\subset\mathcal{H}$ be the code space and let $P$ denote the orthogonal projection onto $\mathcal{C}$.
We consider the normalized maximally mixed state on the code space,
\[
    \sigma := \frac{1}{\dim\mathcal{C}}\,P.
\]
Given the decoder $\mathcal{N}_\theta$ associated with a noise parameter $\theta\ge 0$, we evaluate its performance using the entanglement fidelity
\[
    F_e(\sigma,\mathcal{N}_\theta).
\]
This quantity measures how well quantum correlations encoded in the code space are preserved after the action of noise and decoding.

We define the residual error function by
\[
    T(\theta) := 1 - F_e(\sigma,\mathcal{N}_\theta).
\]
By construction, $T(\theta)=0$ if and only if the entanglement fidelity is equal to one, meaning that entanglement is perfectly preserved by the decoding procedure.

We interpret $T(\theta)$ as an effective error rate after one round of decoding and consider the iterative process
\[
    \theta_{n+1} = T(\theta_n),
\]
starting from an initial error rate $\theta_0$.
If the sequence $\{\theta_n\}$ converges to zero, then repeated decoding asymptotically restores the encoded quantum information.

Based on this observation, we define the threshold as follows.
We say that a value $\theta_\ast>0$ is a threshold if, for all initial error rates $\theta_0\le\theta_\ast$, the iteration satisfies
\[
    \lim_{n\to\infty} T^{(n)}(\theta_0)=0,
\]
where $T^{(n)}$ denotes the $n$-fold composition of $T$.
The threshold is defined as the supremum of all such $\theta_\ast$.

A simple sufficient condition for the existence of a positive threshold can be derived from the small-noise behavior of $T(\theta)$.
Assume that there exist constants $k$ and $\gamma$ such that
\[
    0\le k<1,\qquad \gamma>0,
\]
and that, for sufficiently small $\theta$, the inequality
\[
    T(\theta)\le k\theta+\gamma\theta^2
\]
holds.
Define
\[
    \theta_{\mathrm{th}} := \frac{1-k}{\gamma}.
\]
Then, for any initial value $\theta_0<\theta_{\mathrm{th}}$, one has
\[
    \lim_{n\to\infty} T^{(n)}(\theta_0)=0.
\]

Indeed, in the region $\theta<\theta_{\mathrm{th}}$, the inequality $T(\theta)<\theta$ holds, so the sequence $\{\theta_n\}$ is monotonically decreasing and bounded below by zero.
Therefore, it converges to zero.

This argument shows that the threshold increases as the linear coefficient $k$ decreases.
In particular, suppressing the first-order contribution to the residual error directly enhances the threshold.
This observation provides the conceptual basis for the analysis of threshold enhancement by internal perturbations in the following sections.

\subsection{Leakage to the Outside and the Threshold}

In this subsection we study how ``leakage to the outside'' degrades the one-step performance, using a deliberately poor decoder that randomizes any population that has leaked out of the code space instead of attempting active correction.

Let $\mathcal H$ be a Hilbert space and let $P$ be the orthogonal projection onto the code space
\[
    \mathcal C:=P\mathcal H,\qquad d:=\dim \mathcal C<\infty,\qquad \sigma:=\frac{P}{d}.
\]
We consider a noise channel $\mathcal E_\theta$ with Kraus operators
$\{E_i(\theta)\}_{i\in I}$ such that, as $\theta\to 0$,
\[
    E_0(\theta)=\sqrt{1-\theta}\,I+O(\theta),\qquad
    E_i(\theta)=\sqrt{\theta}\,F_i+O(\theta^{3/2})\quad (i\neq 0),
\]
for some fixed operators $\{F_i\}_{i\neq 0}$.
(We do not assume any correction property here; this subsection is only about estimating the performance of a specific decoder.)

\medskip

\noindent
\textbf{A poor decoder (randomization outside $\mathcal C$).}
Define a CPTP map $\Pi$ by
\begin{equation}\label{eq:poor-decoder}
    \Pi(X):=PXP+\Tr\!\bigl((I-P)X\bigr)\,\sigma.
\end{equation}
Equivalently, choose orthonormal bases $\{|k\rangle\}_{k=1}^d$ of $\mathcal C$ and $\{|\alpha\rangle\}_\alpha$ of $(I-P)\mathcal H$, and set
\[
    K^0:=P,\qquad
    K^{(\mathrm{out})\alpha}:=\frac{1}{\sqrt d}\sum_{k=1}^d |k\rangle\langle \alpha|.
\]
Then
\[
    \Pi(X)=K^0XK^{0*}+\sum_\alpha K^{(\mathrm{out})\alpha}XK^{(\mathrm{out})\alpha*},
\]
and the ``decoded channel'' is
\[
    \mathcal D_\theta:=\Pi\circ \mathcal E_\theta.
\]
We evaluate the one-step failure probability by the entanglement infidelity
\[
    \widetilde T(\theta):=1-F_e(\sigma,\mathcal D_\theta),
    \qquad
    F_e(\sigma,\Lambda)=\sum_j\bigl|\Tr(\sigma L_j)\bigr|^2
    \ \text{ for any Kraus family }\ \{L_j\}\ \text{of }\Lambda.
\]

\medskip

\noindent
\textbf{Leakage probability and variance.}
For $X\in B(\mathcal H)$ define the variance with respect to $\sigma$ by
\[
    Var_\sigma(X):=Tr(\sigma X^*X)-|Tr(\sigma X)|^2.
\]
Define the leakage probability of $\mathcal E_\theta$ (starting from $\sigma$) by
\begin{equation}\label{eq:leak-prob}
    P_\ell(\theta):=\Tr\!\bigl((I-P)\,\mathcal E_\theta(\sigma)\bigr)
    =\sum_{i\in I}\Tr\!\bigl((I-P)E_i(\theta)\sigma E_i(\theta)^*\bigr).
\end{equation}
Using the small-noise expansion above, one has
\begin{equation}\label{eq:leak-prob-expansion}
    P_\ell(\theta)=\theta\sum_{i\neq 0}\Tr\!\bigl((I-P)F_i\sigma F_i^*\bigr)+O(\theta^2).
\end{equation}

\begin{restatable}{lemma}{decoder}
\label{lem:poor-decoder-expansion}
With the definitions above,
    \begin{equation}\label{eq:Ttilde-expansion}
        \widetilde T(\theta)
        \;=\;
        \theta\sum_{i\neq 0}Var_\sigma(PF_iP)
        \;+\;
        \Bigl(1-\frac{1}{d^2}\Bigr)P_\ell(\theta)
        \;+\;
        O(\theta^2).
    \end{equation}
\end{restatable}

The proof of this lemma is given in \ref{app:proof}.

We recall that the one-step failure probability for the poor decoder $\Pi$ is given by
\[
    \widetilde T(\theta)=1-F_e(\sigma,\mathcal D_\theta),
    \qquad
    \mathcal D_\theta=\Pi\circ\mathcal E_\theta,
\]
and that, for small $\theta$, it admits the expansion
\[
    \widetilde T(\theta)
    =
    \theta\sum_i Var_\sigma(PF_iP)
    +
    \Bigl(1-\frac{1}{d^2}\Bigr)P_\ell(\theta)
    +
    O(\theta^2).
\]

We first relate the variance term to the deviation from the Knill--Laflamme condition.
For any operator $X$ and state $\sigma$, the variance can be written as
\[
    Var_\sigma(X)
    =
    \min_{\alpha\in\mathbb{C}}
    Tr\!\bigl(\sigma (X-\alpha I)^*(X-\alpha I)\bigr)
    \le
    \min_{\alpha\in\mathbb{C}}\|X-\alpha I\|^2.
\]
Applying this bound to $X=PF_iP$, we introduce
\[
    \varepsilon_i
    :=
    \inf_{\alpha\in\mathbb{C}}\|PF_iP-\alpha P\|,
\]
which quantifies the extent to which the error operator $F_i$ violates the Knill--Laflamme condition.
Using this definition, we obtain
\[
    \widetilde T(\theta)
    \le
    \theta\sum_i \varepsilon_i^2
    +
    \Bigl(1-\frac{1}{d^2}\Bigr)P_\ell(\theta)
    +
    O(\theta^2).
\]

The quantity $\widetilde T(\theta)$ corresponds to the failure probability of a poor decoder.
Let
\[
    T(\theta):=1-F_e(\sigma,\mathcal N_\theta)
\]
denote the failure probability for the optimal decoder $\mathcal N_\theta$ introduced in Section~4.1. Since the optimal decoder cannot perform worse than $\Pi$, we have
\[
    T(\theta)\le \widetilde T(\theta),
\]
and therefore
\[
    T(\theta)
    \le
    \theta\sum_i \varepsilon_i^2
    +
    \Bigl(1-\frac{1}{d^2}\Bigr)P_\ell(\theta)
    +
    O(\theta^2).
\]
Thus, the true failure probability is controlled by both the deviation from the Knill--Laflamme condition and the leakage probability.

We now bound the leakage probability using spectral properties of the Dirac operator.
Let $(A,\mathcal H,D)$ be a spectral triple satisfying the following assumptions: the operator $D$ has a spectral gap $\Delta>0$,
\[
    \mathrm{spec}(D)\subset (-\infty,-\Delta]\cup\{0\}\cup[\Delta,\infty),
\]
the projection $P$ is the spectral projection onto the zero eigenspace of $D$, and the error operators satisfy
\[
    \|E_i\|\le 1,
    \qquad
    \|[D,E_i]\|\le \varepsilon_i.
\]

Using the contour integral representation of the spectral projection,
\[
    P=\frac{1}{2\pi i}\oint_\Gamma (z-D)^{-1}\,dz,
\]
where $\Gamma$ is a contour enclosing $0$, one obtains the bound
\[
    \|[P,E_i]\|\le C\,\frac{\varepsilon_i}{\Delta}
\]
for a universal constant $C>0$.
It follows that
\[
    \|(I-P)E_iP\|
    \le
    \|[P,E_i]\|
    \le
    C\,\frac{\varepsilon_i}{\Delta}.
\]

Using this estimate, the leakage probability
\[
    P_\ell(\theta)
    =
    \theta\sum_i \Tr\!\bigl((I-P)F_i\sigma F_i^*\bigr)
    +
    O(\theta^2)
\]
can be bounded as
\[
    P_\ell(\theta)
    \le
    \theta\sum_i C^2\,\frac{\varepsilon_i^2}{\Delta^2}
    +
    O(\theta^2).
\]

Combining the above inequalities, we arrive at
\[
    T(\theta)
    \le
    \theta\sum_i \varepsilon_i^2
    +
    \theta\sum_i C^2\,\frac{\varepsilon_i^2}{\Delta^2}
    +
    O(\theta^2)
    =
    k\,\theta
    +
    O(\theta^2),
\]
where
\[
    k:=\sum_i \varepsilon_i^2\Bigl(1+\frac{C^2}{\Delta^2}\Bigr).
\]

As discussed in Section~4.2, the existence of a positive threshold requires $k<1$.
The expression above shows explicitly that increasing the spectral gap $\Delta$ suppresses the leakage contribution and decreases $k$, thereby enhancing the threshold.
This establishes a direct quantitative connection between the spectral gap of the Dirac operator and the error-correction threshold.

In summary, leakage to the outside of the code space is a primary mechanism that limits the threshold.
When the error operators almost commute with the Dirac operator and the spectral gap is large, leakage is strongly suppressed, leading to a significant improvement of the threshold.

\subsection{Threshold Enhancement by Internal Perturbations}

In this subsection, we show that internal perturbations of the Dirac operator can systematically enhance the error-correction threshold.
The key mechanism is that a suitable perturbation preserves the code space while increasing the spectral gap, thereby suppressing leakage and improving the effective contraction rate of the decoding map.

Let $D$ be the Dirac operator defining the code space
\[
    \mathcal C=\ker D,
    \qquad
    P\ \text{the orthogonal projection onto }\mathcal C.
\]
We consider a bounded self-adjoint operator $V$ satisfying the following conditions: $V=V^*$, $PV=VP$, and there exists a constant $c>0$ such that
\[
    (I-P)V(I-P)\ge c\,(I-P).
\]
Such an operator $V$ will be called a \emph{code-preserving perturbation}. The commutation condition $PV=VP$ ensures that the perturbation does not mix the code space with its orthogonal complement, while the positivity condition ensures that $V$ lifts the energy of states outside the code space.

For $\lambda\ge 0$, we define the perturbed Dirac operator by
\[
    D_\lambda:=D+\lambda V.
\]
This is a type of inner fluctuation in noncommutative geometry; see
\ref{app:inner} for details on inner fluctuations.
We first verify that the perturbation preserves the code space.
Indeed, if $v\in\ker D$, then $Dv=0$ and $PV=VP$ imply $(I-P)v=0$ and hence $Vv=0$, so $D_\lambda v=0$.
Conversely, if $v\notin\ker D$, then $(I-P)v\neq 0$, and using the spectral gap $\mu>0$ of $D$ together with the positivity of $(I-P)V(I-P)$, we obtain
\[
    \langle v,D_\lambda v\rangle
    =
    \langle v,Dv\rangle+\lambda\langle v,Vv\rangle
    \ge
    \mu\,\|(I-P)v\|^2+\lambda c\,\|(I-P)v\|^2
    >
    0,
\]
which shows that $v$ cannot belong to $\ker D_\lambda$.
Thus,
\[
    \ker D_\lambda=\ker D.
\]

Next, we examine the effect of the perturbation on the spectral gap.
Assume that the original Dirac operator has a gap $\mu>0$, namely
\[
    \mathrm{spec}(D)\subset(-\infty,-\mu]\cup\{0\}\cup[\mu,\infty).
\]
On the orthogonal complement of the code space, the operator $D_\lambda$ satisfies
\[
    \|(I-P)D_\lambda(I-P)\|
    \ge
    \|(I-P)D(I-P)\|+\lambda\,(I-P)V(I-P)
    \ge
    \mu+\lambda c.
\]
Therefore, the perturbed operator $D_\lambda$ has a spectral gap
\[
    \mu(\lambda)\ge \mu+\lambda c,
\]
which increases linearly with the perturbation strength $\lambda$.

We now relate this gap enhancement to the threshold.
From the analysis in the previous subsection, the one-step failure probability of the optimal decoder satisfies, for small $\theta$,
\[
    T(\theta)\le k\,\theta+O(\theta^2),
    \qquad
    k=\sum_i \varepsilon_i^2\Bigl(1+\frac{C^2}{\Delta^2}\Bigr),
\]
where $\Delta$ denotes the spectral gap of the Dirac operator and
$\varepsilon_i=\|[D,F_i]\|$ measures the non-commutativity of the error operators with $D$.
Replacing $D$ by $D_\lambda$ amounts to replacing $\Delta$ by $\Delta+\lambda c$, and hence
\[
    k(\lambda)
    =
    \sum_i \varepsilon_i^2
    \Bigl(1+\frac{C^2}{(\Delta+\lambda c)^2}\Bigr).
\]
As $\lambda$ increases, $k(\lambda)$ decreases monotonically.

According to the criterion established in Section~4.2, a positive threshold exists whenever $k<1$, and the threshold value increases as $k$ decreases.
Therefore, by tuning the strength $\lambda$ of a code-preserving internal perturbation, one can systematically enhance the error-correction threshold.

In summary, internal perturbations that preserve the code space but lift the energy of states outside it provide a robust and conceptually simple mechanism for threshold enhancement.
This result highlights how noncommutative geometric structure can be used as a control knob for suppressing leakage and optimizing the performance of quantum error-correcting codes.

\subsection{Concrete example: three-qubit repetition code}
\label{subsec:3qubit-example}

To make the role of the spectral gap in suppressing leakage more transparent, we illustrate the mechanism using the simplest nontrivial stabilizer code: the three-qubit repetition code.

The physical Hilbert space is $\mathcal H = (\mathbb C^2)^{\otimes 3}$.
The stabilizer generators are
\begin{equation}
    S_1 = Z_1 Z_2, \qquad S_2 = Z_2 Z_3 ,
\end{equation}
and the code space is
\begin{equation}
    \mathcal C = \{ \psi \in \mathcal H \mid S_1 \psi = \psi,\; S_2 \psi = \psi \}
    = \mathrm{span}\{ |000\rangle, |111\rangle \}.
\end{equation}

In the spectral code framework, we define the Dirac operator by
\begin{equation}
    D = (1 - Z_1 Z_2) + (1 - Z_2 Z_3).
\end{equation}
Its kernel coincides with the code space, $\ker D = \mathcal C$.
States violating one stabilizer have eigenvalue $2$, while states violating both stabilizers have eigenvalue $4$.
Hence the spectral gap separating the code space from excited states is
\begin{equation}
    \Delta = 2 .
\end{equation}

Consider a local Pauli error $F = X_1$.
Acting on a code state, it produces a stabilizer violation:
\begin{equation}
    X_1 |000\rangle = |100\rangle, \qquad
    D |100\rangle = 2 |100\rangle .
\end{equation}
Thus, a single-qubit error necessarily excites the system across the gap $\Delta$.
We model the noise by the channel
\begin{equation}
    \mathcal E_\theta(\rho) = (1-\theta)\rho + \theta X_1 \rho X_1 ,
\end{equation}
which is consistent with the small-noise expansion assumed in Section~4.

We now introduce a code-preserving internal perturbation
\begin{equation}
    D_\lambda = D + \lambda (I - P),
\end{equation}
where $P$ is the projection onto $\mathcal C$.
This perturbation leaves the code space invariant, $\ker D_\lambda = \mathcal C$, while increasing the spectral gap to
\begin{equation}
    \Delta(\lambda) = 2 + \lambda .
\end{equation}

As shown in Section~4.3, the leakage amplitude induced by a local error operator is bounded by the commutator with the Dirac operator and suppressed by the spectral gap.
In the present example, this implies
\begin{equation}
    \|(I-P) X_1 P\| \;\lesssim\; \frac{1}{\Delta(\lambda)},
\end{equation}
and hence the leakage probability scales as
\begin{equation}
    P_\ell(\theta) \sim \frac{\theta}{(\Delta+\lambda)^2}.
\end{equation}
Therefore, increasing the spectral gap by an internal perturbation monotonically suppresses leakage out of the code space and improves the effective error-correction threshold.

This simple stabilizer-code example provides a concrete physical realization of the general mechanism discussed in Sections~4.3 and~4.4: local errors necessarily create energetic excitations, and enlarging the spectral gap raises the energetic cost of such leakage without modifying the encoded logical information.

\section{Strict Deformation Quantization and Spectral Codes}

In this section, we demonstrate that Berezin--Toeplitz quantization, which is a well-known example of strict deformation quantization, admits a natural interpretation within the framework of spectral codes.
This provides a concrete realization of spectral codes in which classical geometric data are systematically deformed into quantum error-correcting structures.

Section~5.1 introduces the notion of strict deformation quantization. In Section~5.2, matrix regularization is discussed as a representative example. Section~5.3 presents Berezin--Toeplitz quantization as a concrete realization of matrix regularization.
In Section~5.4, we explain how Berezin--Toeplitz quantization can be regarded as a spectral code. Finally, Section~5.5 comments on possible physical implementations of this construction.

\subsection{Definition of Strict Deformation Quantization}

We begin by recalling the definition of strict deformation quantization.
For this purpose, we first introduce the notion of a continuous field, which will play an important role in what follows \cite{natsume}.

\begin{definition}[continuous field of Banach spaces]
    Let $\Omega$ be a topological space, and let
    $\mathcal{B}_\Omega = (B_{\omega})_{\omega \in \Omega}$ be a family of Banach spaces.
    Let $\mathbb{B}$ be a linear subspace of the product
    \[
        \prod B_{\omega}
        =
        \{ (x(\omega))_{\omega\in\Omega} \mid x(\omega)\in B_{\omega} \text{ for all }
        \omega\in\Omega \}.
    \]
    The pair $(\mathcal{B}_\Omega, \mathbb{B})$ is called a continuous field of Banach spaces if the following conditions are satisfied:
    \begin{enumerate}
        \item For every $\omega\in\Omega$, the set
        $\{x(\omega)\mid x\in\mathbb{B}\}$ is dense in $B_{\omega}$.
        \item If $x\in\prod B_{\omega}$ and for any $\omega\in\Omega$ and
        $\epsilon>0$ there exists $x'\in\mathbb{B}$ such that
        \[
            \|x(\omega)-x'(\omega)\|\le\epsilon,
        \]
        then $x$ belongs to $\mathbb{B}$.
        \item For each $x\in\mathbb{B}$, the map
        \[
            \omega\longmapsto \|x(\omega)\|
        \]
        is continuous.
    \end{enumerate}
\end{definition}

This structure may be viewed as an analogue of a fiber bundle whose fibers are Banach spaces and whose base space is $\Omega$.
In this picture, the subspace $\mathbb{B}$ plays the role of a space of continuous sections.
We next specialize this notion to the case of C$^\ast$-algebras.

\begin{definition}[continuous field of C$^\ast$-algebras]
    If each $B_{\omega}$ in a continuous field of Banach spaces is a C$^\ast$-algebra and $\mathbb{B}$ is a $^\ast$-subalgebra of $\bigoplus B_{\omega}$, then $\mathbb{B}\subset \prod B_{\omega}$ is called a continuous field of C$^\ast$-algebras.
\end{definition}

This definition can be regarded as a C$^\ast$-algebraic counterpart of a fiber bundle, with $\mathbb{B}$ playing the role of the space of sections.
Using this framework, we now define strict deformation quantization \cite{rieffel1993deformation,hawkins1999quantization,hawkins2000geometric}.

\begin{definition}[Strict Deformation Quantization]
    Let $M$ be a Poisson manifold.
    A strict deformation quantization of $M$ consists of a continuous field of C$^\ast$-algebras $( A_{\hat{\mathcal{J}}}=(A_i)_{i\in\hat{\mathcal{J}}},\mathbb{A})$ together with a quantization map $Q:C(M)\to\mathbb{A}$.
    Here $\hat{\mathcal{J}}\subset \mathbb{R}_{\ge0}\cup\{\infty\}$ is an index set containing the classical limit $\infty$, for which $A_{\infty}=C(M)$.
    Let $\mathbb{A}$ denote the algebra of continuous sections of
    $ A_{\hat{\mathcal{J}}}$, and let $\mathcal{P}:\mathbb{A}\to C(M)$ be the evaluation map at $\infty$.
    The quantization map $Q$ is required to satisfy
    \[
        \mathcal{P}\circ Q = \mathrm{id},
    \]
    that is, $Q_{\infty}=\mathrm{id}$.
    Writing $Q_i$ for the component of $Q$ at $i\in\mathcal{J}
    =\hat{\mathcal{J}}\setminus\{\infty\}$, we further require
    \[
        \lim_{i\to\infty}\bigl\| i\hbar_i[Q_i(f),Q_i(g)]-Q_i(\{f,g\})\bigr\|=0,
    \]
    where $\hbar_i$ is a monotonically decreasing function of $i$.
\end{definition}

The above definition can equally well be formulated for $C^\infty(M)$ instead of $C(M)$.
In many cases, the map $Q$ can be reconstructed from the family $\prod_{i\in\mathcal{J}} Q_i$, and the convergence conditions ensure that $\prod_{i\in\mathcal{J}}\mathrm{Im}\,Q_i=\mathbb{A}$.

We now specialize to the case
\[
    \hat{\mathcal{J}}=\hat{\mathbb{N}}=\{1,2,\dots,\infty\},
    \qquad
    A_p\ (p\in\mathbb{N})\ \text{matrix algebras},
    \qquad
    M\ \text{a connected, closed K\"ahler manifold}.
\]
Let $\mathbb{A}_0=\ker\mathcal{P}$ be the ideal of sections vanishing at $\infty$.
Since $\mathbb{N}$ is discrete, one has
\[
    \mathbb{A}_0=\bigoplus_{p\in\mathbb{N}}A_p
    =\{(a,p)\mid p\in\mathbb{N},\,a\in A_p\}.
\]
This yields the short exact sequence
\[
    \xymatrix{
    0 \ar[r] & \mathbb{A}_0 \ar[r] & \mathbb{A} \ar[r]^-{\mathcal{P}} & C(M) \ar[r] & 0 }.
\]

\subsection{Example: Matrix Regularization}

We next recall the notion of matrix regularization from a mathematical viewpoint.
Let $(M,\omega)$ be a symplectic manifold of dimension $2n$, where $\omega$ is a non-degenerate closed two-form.
The symplectic form defines the volume form $\mu_{\omega}:=\omega^{\wedge n}/n!$ and the Poisson bracket
$\{f,g\}=X_f g$, where $X_f$ denotes the Hamiltonian vector field associated with $f$.
Thus, symplectic manifolds provide a natural generalization of classical phase spaces.
We assume throughout that $M$ is closed, i.e., compact and without boundary.

A matrix regularization of $M$ is given by a sequence of linear maps $T_p:C^\infty(M,\mathbb{C})\to M_{N_p}(\mathbb{C})$ satisfying the axioms \cite{Arnlind:2010ac,Hoppe}:
\begin{align}
    \label{mat reg 1}
    &\lim_{p\to\infty}\|T_p(f)T_p(g)-T_p(fg)\|=0,\\
    \label{mat reg 2}
    &\lim_{p\to\infty}\|(\im\hbar_p)^{-1}[T_p(f),T_p(g)]-T_p(\{f,g\})\|=0,\\
    \label{mat reg 3}
    &\lim_{p\to\infty}(2\pi\hbar_p)^n\mathrm{Tr}\,T_p(f)
    =\int_M\mu_{\omega}f,\\
    \label{mat reg 4}
    &\lim_{p\to\infty}\|T_p(f)\|=\|f\|.
\end{align}
Here $p\in\mathbb{N}$, $\{N_p\}$ is a strictly increasing sequence of integers, $\hbar_p=(2\pi p)^{-1}$, and $\|\cdot\|$ denotes the operator norm.

\begin{definition}
    Let $T:C(M)\to\prod_{p\in\mathbb{N}}A_p$ be the map defined as the direct product
    of the maps $T_p$.
    We define
    \[
        \mathbb{A}_0=\bigoplus_{p\in\mathbb{N}}A_p,
        \qquad
        \mathbb{A}=\mathbb{A}_0+\mathrm{Im}\,T.
    \]
\end{definition}

While the direct sum of C$^\ast$-algebras is again a C$^\ast$-algebra, the direct product need not be.
Nevertheless, the image of $C(M)$ under $T$, together with the direct sum $\mathbb{A}_0$, generates a well-defined C$^\ast$-algebra.
One has the following result.

\begin{theorem}
    The algebra $\mathbb{A}$ is a C$^\ast$-algebra, and the map $T$ induces an isomorphism
    \[
        C(M)\cong \mathbb{A}/\mathbb{A}_0.
    \]
\end{theorem}

Hence, via the canonical projection, one obtains
\[
    \mathbb{A}\longrightarrow \mathbb{A}/\mathbb{A}_0\cong C(M),
    \qquad
    \mathcal{P}:\mathbb{A}\to C(M),
\]
which shows that matrix regularization naturally gives rise to a strict deformation quantization \cite{hawkins1999quantization,hawkins2000geometric}.

\begin{theorem}
    With respect to the continuous field of C$^\ast$-algebras over
    $\hat{\mathbb{N}}$, whose fibers are $A_N$ for $N\in\mathbb{N}$ and $C(M)$ at $\infty$, the map $\mathcal{P}$ defines a strict deformation quantization.
\end{theorem}

\subsection{Construction using Berezin--Toeplitz Quantization}

In this subsection, we review Berezin--Toeplitz quantization in the spinorial formulation (allowing auxiliary vector bundles) and record the asymptotic behavior of the resulting quantization map.
Berezin--Toeplitz quantization was originally introduced in two dimensions in \cite{1992LMaPh..24..125K}, later extended to general K\"ahler manifolds in \cite{Bordemann:1993zv}, and further generalized to symplectic manifolds in \cite{ma2008toeplitz}.
While it was first applied mainly to functions via matrix regularization, it has since been developed for vector bundles as well \cite{hawkins2000geometric,Adachi:2021aux,Adachi:2021ljw,Adachi:2022mln,BTQtextbook,Adachi:2023dhc}.

\subsubsection{Setup}

Let \(M\) be a closed, connected \(2n\)-dimensional K\"ahler manifold equipped with a K\"ahler structure \((g,J,\omega)\), where \(g\) is a Riemannian metric, \(J\) is a complex structure, and \(\omega\) is a symplectic form satisfying
\begin{align}
    \label{Kahler}
    \omega(\cdot,\cdot)=g(J\cdot,\cdot).
\end{align}
Locally one may introduce a K\"ahler potential \(K\) such that
\(\omega=\im\partial\bar{\partial}K\), where \(\partial,\bar{\partial}\) are the Dolbeault operators.
The natural volume form is \(\mu:=\omega^{\wedge n}/n!\), which in local real coordinates \(\{x^\mu\}_{\mu=1}^{2n}\) takes the form
\begin{align}
    \mu=\sqrt{g}\,\dif x^1\wedge\dif x^2\wedge\cdots\wedge\dif x^{2n}.
\end{align}

To define the quantization map we introduce two Hermitian bundles \(L\) and \(S_c\), where \(L\) is a prequantum line bundle and \(S_c\) is a spin-\(c\) bundle.
For an arbitrary vector bundle \(F\), we denote its connection by
\(\nabla^F=\dif+A^F\) and its curvature by \(R^F:=(\nabla^F)^2=\dif A^F + A^F\wedge A^F\).

A prequantum line bundle \(L\) is a complex line bundle equipped with a connection \(\nabla^L\) whose curvature satisfies
\[
    R^L=-\im k\,\omega.
\]
Here \(k\) is chosen so that \(\frac{\im}{2\pi}\int_\Sigma R^L\) is an integer for any two-cycle \(\Sigma\subseteq M\), equivalently
\(\frac{k}{2\pi}\omega\in H^2(M,\mathbb{Z})\).
Manifolds admitting such an \(L\) are called quantizable.
In two dimensions, for \(M=\Sigma\), one may take \(k=2\pi/\int_M\omega\). Writing \(\nabla^L=\dif+A^L\) locally and using the K\"ahler potential \(K\), one obtains
\[
    A^L=-\frac{k}{2}(\partial-\bar{\partial})K.
\]
Let \(\Gamma(\cdot)\) denote the space of smooth sections.
Thus \(\Gamma(L)\) can be viewed as complex scalar fields coupled to a \(\mathrm{U}(1)\) background gauge field \(A^L\).
In two dimensions, the fact that the curvature is proportional to the volume form indicates that sections of \(L\) behave as scalar fields in a uniform magnetic flux background.

We next recall the spin-\(c\) structure (see \cite{BTQtextbook, spingeometry} for a more detailed account).
Not every K\"ahler manifold is spin, i.e.\ it may not admit a spin bundle \(S\).
Nevertheless, every K\"ahler manifold (more generally, every almost complex manifold) admits a spin-\(c\) structure.
The canonical spin-\(c\) bundle is
\[
    S_c:=\bigoplus_{p=0}^n \Lambda^{0,p}(T^\ast M),
\]
so its fibers consist of \((0,p)\)-forms.
Formally one may write \(S_c\simeq S\otimes L_c^{1/2}\), where \(L_c\) is the determinant line bundle of the holomorphic tangent bundle, \(L_c:=\det T^{(1,0)}M\).
For non-spin manifolds, \(S\) and \(L_c^{1/2}\) are not defined individually, but the product \(S_c=S\otimes L_c^{1/2}\) is well-defined \footnote{$CP^{2m}\ (m\in\mathbb{N})$ is an example of a non-spin manifold with the spin-\(c\) structure.}.
A connection on \(S_c\) is locally written as \(\nabla^{S_c}=\dif + A^S + \frac12 A^{L_c}\), where the connection one-form of the canonical spin bundle \(S\) is
\begin{align}
    A^S = \frac{1}{4}\gamma_{(2n)}^a\gamma_{(2n)}^b\Omega_{ab},
\end{align}
with \(\{\gamma_{(2n)}^a\}_{a=1}^{2n}\) satisfying \(\{\gamma_{(2n)}^a,\gamma_{(2n)}^b\}=2\delta_{ab}I_{2^n}\) (see
Appendix~\ref{Clifford}).
Here \(\Omega_{ab}=\Omega_{ab\mu}\dif x^\mu\) is the spin connection one-form:
\begin{align}
    \Omega_{ab\mu} = e_a{}^{\nu} g_{\nu \rho}
    (\partial_{\mu} e_{b}{}^{\rho} + \Gamma^{\rho}_{\mu \sigma} e_b{}^{\sigma}).
\end{align}
The fields \(\{e_a\}_{a=1}^{2n}\) are local orthonormal frames (vielbeins) with \(g(e_a,e_b)=\delta_{ab}\).
The connection one-form of \(L_c\) is \(A^{L_c}=-\sum_{m=1}^n\Omega_{m\bar m}\), where \(m,\bar m\) label complexified orthonormal frame vectors introduced in Appendix~\ref{complex orthonormal}.
Accordingly, sections of \(S_c\) may be interpreted as complex spinor fields coupled to \(\frac12 A^{L_c}\).

\subsubsection{Definition of Toeplitz operator}

As noted above, any \(f\in C^\infty(M)\) acts fiberwise on
\(\Gamma(S_c\otimes L^{\otimes p})\) by multiplication, hence can be regarded as a linear map
\[
    \Gamma(S_c\otimes L^{\otimes p}) \longrightarrow \Gamma(S_c\otimes L^{\otimes p}),
\]
where \(p\in\mathbb{Z}\) is fixed.
These spaces are infinite-dimensional.
The basic idea of Berezin--Toeplitz quantization is to restrict such maps to a finite-dimensional subspace so that they can be represented by finite matrices.

To this end, consider the Dirac operator acting on
\(\Gamma(S_c\otimes L^{\otimes p})\),
\begin{align}
    \label{twisted Dirac}
    D = \im \gamma_{(2n)}^{a} \nabla^{S_c\otimes L^{\otimes p} }_{e_a}
    = \im e_a{}^{\mu}\gamma_{(2n)}^{a} \left(
    \partial_{\mu} + \frac{1}{4} \Omega_{ab\mu} \gamma_{(2n)}^a \gamma_{(2n)}^b
    - \frac{1}{2} \sum_{m=1}^n \Omega_{m \bar{m} \mu}
    + p A^L_{\mu} \right).
\end{align}
We equip \(\Gamma(S_c\otimes L^{\otimes p})\) with the inner product
\begin{align}
    \label{inner product}
    (\psi',\psi) := \int_M \mu \, (\psi')^{\dagger} \cdot \psi
    \quad (\psi,\psi' \in \Gamma(S_c \otimes L^{\otimes p})),
\end{align}
where \((\psi')^\dagger\cdot\psi\) denotes the Hermitian fiber inner product coming from the Hermitian metrics on \(S_c\) and \(L\).
We write \(|\psi|=\sqrt{(\psi,\psi)}\).

The space of normalizable zero modes \(\ker D\) is finite-dimensional.
With the gamma-matrix convention in Appendix~\ref{Clifford}, one can compute \(N:=\dim\ker D\) for sufficiently large \(p\) using the Atiyah--Singer index theorem together with a vanishing theorem (see Appendix~\ref{Vanishing}).
Thus \(p\) controls the dimension \(N\), which plays the role of the matrix size in the resulting matrix regularization.

Let \(\Pi\) be the orthogonal projection from \(\Gamma(S_c\otimes L^{\otimes p})\) onto \(\ker D\).
The Berezin--Toeplitz quantization map is then defined by
\begin{align}
    T_p(f) = \Pi f \Pi \quad (f \in C^\infty(M)).
    \label{Toeplitz op}
\end{align}
Since \(T_p(f)\) maps \(\ker D\) to itself, it can be represented as an \(N\times N\) matrix.
As we explain below, the family \(\{T_p(f)\}\) admits a useful asymptotic expansion, which yields a generalization of matrix regularization.

\subsubsection{Asymptotic Expansion of Toeplitz Operator}

Consider the product \(T_p(f)T_p(g)\).
As established in Appendix~\ref{asymptotic expansion}, the Toeplitz operators \eqref{Toeplitz op} satisfy an asymptotic expansion in
\(\hbar_p=(kp)^{-1}\) of the form
\begin{align}
     T_p(f) T_p(g) = \sum_{i=0}^{\infty} \hbar_{p}^i T_p(C_i(f,g)).
    \label{asym exp}
\end{align}
Here each \(C_i\) is a bidifferential operator \(C^\infty(M)\times C^\infty(M)\to C^\infty(M)\).
The first three coefficients are given explicitly by
\begin{align}
    C_0(f,g) &= fg,\\
    C_1(f,g) &= -\frac{1}{2} G^{\alpha\beta}(\partial_{\alpha}f)(\partial_{\beta}g),\\
    C_2(f,g) &= \frac{1}{8} G^{\alpha\beta}G^{\gamma\delta}
    \bigl[(\partial_{\alpha} f)
    ( \im R_{\beta \gamma \mu \nu} W^{\mu\nu} )
    (\partial_{\delta}g)
    + (\partial_{\alpha} \partial_{\gamma} f)
    (\partial_{\beta} \partial_{\delta} g)\bigr].
    \label{asymptotic exp}
\end{align}
We have introduced
\(G^{\alpha\beta}:=g^{\alpha\beta}+\im W^{\alpha\beta}\), where \(g^{\alpha\beta}\)
is the inverse metric and \(W^{\alpha\beta}\) is the Poisson tensor defined by
\(\omega_{\mu\nu}W^{\mu\rho}=\delta_\nu^\rho\).
In \eqref{asymptotic exp}, \(R_{\alpha\beta\gamma\delta}\) denotes the Riemann curvature tensor of \(g\).

We refer to Appendix~\ref{asymptotic expansion} for the proof of
\eqref{asym exp}, and we record here useful consequences.
Equation~\eqref{asym exp} readily implies
\begin{align}
    \label{generalized matrix regularization}
    &\lim_{p\to\infty}\left\|T_p(f)T_p(g) - T_p(fg)\right\|=0.\\
    &\lim_{p\to\infty} \left\|\im \hbar_{p}^{-1}[T(f) , T_p(g)] -T_p(\{f,g\})\right\|=0,
\end{align}
which reproduce the first two axioms of matrix regularization in the Berezin--Toeplitz setting.

\subsubsection{Trace of Toeplitz Operator}

We also consider the trace of Toeplitz operators.
As shown in Appendix~\ref{trace}, one finds
\begin{align}
    \label{trace property}
    \lim_{p \to \infty} (2\pi\hbar_p)^n \Tr T_p(f) = \int_M  \mu.
\end{align}
This reproduces the third axiom of matrix regularization.
Furthermore, one can prove (see \cite{hawkins2000geometric}) that Berezin--Toeplitz quantization satisfies
\begin{align}
    \label{abs property}
    \lim_{p\to \infty} \|T_p(f)\| = \| f\| .
\end{align}
Together with \eqref{generalized matrix regularization} and \eqref{trace property}, this shows that Berezin--Toeplitz quantization fulfills all four axioms of matrix regularization, including \eqref{abs property}.
In particular, it provides a concrete and canonical realization of matrix regularization in which the geometry of the underlying K\"ahler manifold is encoded in a sequence of finite-dimensional matrix algebras.

\subsection{Berezin--Toeplitz Quantization as Quantum Error Correction}

In this subsection we interpret Berezin--Toeplitz (BT) quantization, formulated in terms of spinors, as a concrete example of a spectral code.
The construction of BT quantization itself has been described in the previous section and will be taken as given here.

Let $M$ be a compact spin$^c$ manifold equipped with a compatible K\"ahler structure, and let $S$ denote the spinor bundle.
For each quantization parameter $p$, we consider the twisted Dirac operator $D$ acting on $L^2(M,S\otimes L^p)$, whose kernel
\[
    \mathcal H_p := \ker D
\]
plays the role of the lowest-energy subspace.
We denote by
\[
    \Pi : L^2(M,S\otimes L^p) \to \mathcal H_p
\]
the orthogonal projection onto $\mathcal H_p$.

For a smooth function $f\in C^\infty(M)$, the Berezin--Toeplitz operator is given by
\[
    T_p(f) := \Pi\, M_f\, \Pi,
\]
where $M_f$ denotes multiplication by $f$ on spinors.
In this framework, the triple
\[
    (\Pi C^\infty(M)\Pi,\mathcal H_p,0)
\]
may be viewed as a (degenerate) code spectral triple, with $\Pi$ as the code projection and $\mathcal H_p$ as the code space.

A key feature of this example is that the error algebra is the commutative algebra $C^\infty(M)$.
Thus, quantum information encoded in the spinorial Hilbert space $\mathcal H_p$ is protected against errors generated by classical data, namely functions on the manifold.
In this sense, Berezin--Toeplitz quantization realizes a mixed spectral code, in the terminology of Fig.~\ref{fig:real}, in which quantum information is protected by classical geometric information.

Error operators are given by $E_f := T_p(f)$.
Locality of errors is measured by the support of the function $f$.
When the support of $f$ is sufficiently small compared to the code distance $d_{D}(\Pi)$ determined by the Dirac operator, the corresponding error is correctable.
Heuristically, this follows from the Toeplitz product formula, which approximates multiplication after compression.
Indeed, for any $f,g\in C^\infty(M)$,
\[
    \Pi\, T_p(f)^* T_p(g)\, \Pi=
    \Pi\, T_p(\bar f g)\, \Pi
    +
    O(p^{-1}).
\]
In particular, if $\mathrm{supp}(f)\cap\mathrm{supp}(g)=\emptyset$, then $\bar f g=0$ and
\[
    \Pi\, T_p(f)^* T_p(g)\, \Pi
    =
    O(p^{-1}),
\]
which shows that the Knill--Laflamme condition is satisfied asymptotically.
Hence, errors associated with functions of sufficiently small support are correctable.

Moreover, the same product formula implies that for arbitrary functions $f$ and $g$,
\[
    \Pi\, T_p(f)^* T_p(g)\, \Pi
    =
    \lambda_{f,g}\,\Pi + O(p^{-1}),
\]
for some scalar $\lambda_{f,g}$ depending on $f$ and $g$.
Therefore, even for global functions, the Knill--Laflamme condition holds approximately in the large-$p$ limit.
This shows that Berezin--Toeplitz quantization provides an approximate quantum error-correcting code for arbitrary classical observables.

As a consequence, the set of correctable errors in this model is determined by the geometric structure of the manifold $M$.
The size and shape of the support of functions, as measured by the Riemannian distance on $M$, control error correctability.
In particular, local geometric features govern exact correctability, while global geometric features control the rate at which approximate correction emerges as $p$ increases.

In summary, Berezin--Toeplitz quantization in the spinorial formulation defines a spectral code in which quantum information is protected against classical errors originating from functions on the manifold. 
The correctability of errors is dictated by the geometry of $M$, providing a clear geometric interpretation of quantum error correction within the framework of noncommutative geometry.

\subsection{Physical implementation via Landau-level projection (Berezin--Toeplitz quantization)}

Berezin--Toeplitz (BT) quantization can be understood as a projection onto the low-energy subspace of charged fermions in a (sufficiently strong) uniform magnetic field.
Concretely, the quantization of an observable $f$ is represented (at the level of operators) in the form
\begin{equation}
    T(f)\;\sim\; P_{\mathrm{LLL}}\, f\, P_{\mathrm{LLL}},
\end{equation}
where $P_{\mathrm{LLL}}$ denotes the projector onto the lowest Landau level (LLL).
In this sense, BT quantization is equivalent to computing the LLL (or, more generally, Landau-level projectors) of a quantum Hall type system, and it provides the natural operator algebra acting on the corresponding low-energy Hilbert space.

In our framework, the code space is identified with such a BT-quantized low-energy subspace.
Therefore, if one can physically realize the LLL (i.e., the relevant low-energy projection) on a chosen background manifold---here ``curved space'' is used simply to indicate that we allow general base manifolds beyond flat space, in order to construct BT quantization on various geometries---then the corresponding quantum error-correcting code can, in principle, be implemented directly at the level of the underlying physical system.

As concrete candidate platforms for realizing (or emulating) Landau-level physics and the associated low-energy projection, one may consider, for example:
\begin{itemize}
  \item \textbf{Strain-engineered graphene (2D electronic systems):}
  strain-induced pseudo-gauge fields can generate (pseudo-)Landau levels and an
  effective LLL structure \cite{guinea2010energy,vozmediano2010gauge,low2010strain}.
  \item \textbf{Topological photonics (resonator arrays / photonic crystals):}
  synthetic gauge fields enable quantum-Hall-like band structures and even synthetic
  Landau levels for photons \cite{ozawa2019topological,hafezi2011robust}.
  \item \textbf{neutral atoms (artificial gauge fields / optical lattices):}
  laser-induced artificial gauge fields allow neutral atoms to mimic charged particles in a magnetic field, giving rise to Landau-level-like spectra or topological bands.
  A well-isolated lowest level or band then provides an effective low-energy projection analogous to the lowest Landau level \cite{dalibard2011colloquium}.
\end{itemize}

\section{Holography and spectral codes}

In this section, we discuss the relation between the holographic principle \cite{hooft1993dimensional,maldacena1999large,witten1998anti,ryu2006holographic} and the spectral codes introduced in this paper.

The holographic principle is an important subject in particle theory and provides a strong clue for understanding quantum gravity, which is one of the ultimate goals of particle physics.

Roughly speaking, it claims that a quantum theory of gravity in a $(d+1)$-dimensional bulk is equivalent to a $d$-dimensional quantum field theory living on its boundary.
A well-known example is the AdS/CFT correspondence, which states that quantum gravity on $(d+1)$-dimensional AdS spacetime is equivalent to a $d$-dimensional conformal field theory on its conformal boundary.
Understanding such a correspondence is very important for probing the properties of quantum gravity in an indirect way.

However, it is still difficult to treat a full quantum theory of gravity directly, so in this section we first focus on the geometry of bulk spacetime described by classical gravity.
We then summarize several characteristic properties that hold between bulk classical gravity and the boundary quantum field theory.

Typically, the following structures are known:
\begin{enumerate}
  \item The bulk gravity theory has local gauge symmetries (diffeomorphism invariance, local Lorentz symmetry, etc.) and therefore has a huge redundancy.
        Many different field configurations represent the same physical spacetime geometry.
  \item Black holes emit radiation via Hawking radiation, and from the holographic point of view it is believed that, in principle, the original state can be reconstructed from the information carried by this radiation.
        This suggests that bulk degrees of freedom are redundantly encoded in many boundary degrees of freedom (see \cite{Hawking:1975vcx,PhysRevLett.71.3743,hayden2007black}).
  \item Local bulk operators can be reconstructed from many boundary degrees of freedom, and in fact in several different ways.
        For example, an operator supported in a certain bulk region can be reconstructed as CFT operators defined on different boundary subsystems (see \cite{hamilton2006local,almheiri2015bulk,jafferis2016relative,dong2016reconstruction,harlow2018tasi}).
\end{enumerate}

Putting these points together, we see that the bulk classical gravity theory (more precisely, the low-energy effective degrees of freedom in the bulk) is redundantly encoded in many degrees of freedom of the boundary field theory, and can still be reconstructed even if part of the boundary degrees of freedom is lost.

This structure is precisely the essence of a quantum error-correcting code, and we can interpret it as saying that
``bulk gravitational degrees of freedom (spacetime geometry) are encoded as a quantum error-correcting code in the boundary field theory'' (see \cite{pastawski2015holographic}).
From this point of view, the holographic principle is deeply connected to quantum error correction, and our spectral code framework may give a way to understand holography.

We now rephrase the above discussion in terms of spectral triples.
First, we assume that the full quantum gravity theory on the boundary side is a large quantum field theory whose Hilbert space contains all relevant states, and that the whole theory can be described by a large noncommutative spectral triple $(A,H,D)$.
Here $A$ is the operator algebra of the boundary CFT, $H$ is its Hilbert space, and $D$ is a self-adjoint operator, such as an energy operator or a Dirac-type operator, that reflects both the spectral structure and geometric information.
This assumption is schematic and serves only to indicate how holography may be organized in the language of noncommutative geometry.
We can regard this noncommutative spectral triple as a description of a full ``noncommutative spacetime'' including quantum gravity effects.

On the other hand, the bulk spacetime geometry described by classical gravity corresponds to the low-energy effective degrees of freedom of this large theory.
Following the spectral code framework, we choose a low-energy interval $I$ in the spectrum of $D$ and define
\[
    P := \chi_I(D), \qquad H_{\mathrm{code}} := P H
\]
so that the low-energy subspace $H_{\mathrm{code}}$ is extracted as a code space.

The compressed spectral triple $(PAP,PH,PDP)$ then defines a spectral code and at the same time gives a new spectral triple describing the low-energy effective theory.
According to the holographic picture, this low-energy effective theory corresponds to the semiclassical bulk spacetime geometry described by classical gravity.

The important point here is that the original algebra $A$, as the operator algebra of the full boundary theory including quantum gravity, is essentially noncommutative, whereas the compressed algebra $PAP$ obtained by projecting to low energy is expected to behave, in a suitable limit, almost as a commutative algebra.
In other words, we are in a situation where $A$ is noncommutative and $PAP$ is (effectively) commutative.

As discussed in the general theory of spectral codes, this is a special case in which a noncommutative algebra $A$ gives rise, after projection, to a commutative compressed algebra $PAP$ as the effective algebra on the code space.
Since commutative spectral triples are equivalent to Riemannian manifolds, the triple $(PAP,PH,PDP)$ in this case describes the bulk spacetime geometry as a Riemannian manifold.

Equivalently, the same projection simultaneously produces (i) a protected code subspace $PH$ and (ii) an approximately classical geometric description encoded by the commutative algebra $PAP$.
In other words, the commutative spectral triple $(PAP,PH,PDP)$ obtained by projecting the noncommutative spectral triple $(A,H,D)$ that includes quantum gravity down to low energies has a double meaning: it defines a code space as a quantum error-correcting code, and at the same time describes spacetime geometry in classical gravity.

From this point of view, the spectral code framework introduced in this paper suggests a way to understand the relation between bulk spacetime geometry (commutative spectral triples) and the boundary quantum theory (noncommutative spectral triples) as a spectral code with noncommutative $A$ and commutative $PAP$.
That is, the holographic principle may be reinterpreted as saying that, starting from a large noncommutative spectral triple, a projection to low energies produces a commutative spectral triple describing spacetime geometry, and this projection has the structure of a quantum error-correcting code.

From the perspective of entanglement wedge reconstruction, the low-energy projection $P=\chi_I(D)$ can be viewed as selecting a semiclassical code subspace on which bulk effective operators admit multiple boundary representations.
Boundary subalgebras associated with different regions may then act identically on $PH$ whenever they reconstruct the same bulk operator inside the corresponding entanglement wedge, reflecting the standard operator-algebraic QEC structure.
Formulating this correspondence purely in spectral terms—by relating the choice of $I$ and the induced Connes geometry on $PH$ to wedge inclusion and complementary recovery—remains an interesting direction for future work.

\section{Conclusion}\label{sec:conclusion}

In this paper, we introduced the notion of \emph{spectral codes}, a framework that unifies quantum error correction and geometry through the language of spectral triples.
The central idea is to regard a low-energy spectral projection of a Dirac-type operator as a code projection, so that the corresponding low-energy subspace naturally defines a quantum error-correcting code.

We first developed the general theory of spectral codes, showing that the code space is spanned by the zero modes of the Dirac operator, which represent global degrees of freedom insensitive to the local metric structure.
This provides a precise geometric mechanism by which logical information is delocalized over the entire space.
As a consequence, correctable errors are characterized as operators that are local with respect to the metric induced by the Dirac operator, while logical information remains protected against such local perturbations.

Within this framework, the Knill--Laflamme condition admits a natural geometric interpretation: sufficiently local errors act trivially on the zero-mode sector, whereas only nonlocal operators can distinguish and therefore act nontrivially on the code space.
From this perspective, the code distance is identified as a geometric quantity determined by the spectral properties of the Dirac operator, in direct analogy with the role of distance in conventional geometric settings.

We then demonstrated that a wide class of known quantum error-correcting codes arise naturally as spectral codes.
Classical linear codes, stabilizer codes, GKP-type codes, and topological codes were all shown to fit into this framework, with their distances reproduced by the same spectral notion.
In each case, the algebraic structure of the code is encoded in the choice of Dirac operator and the associated low-energy projection, showing that apparently disparate code constructions admit a unified geometric description.

A key result of this work is the analysis of decoding and thresholds from a spectral perspective.
We showed that leakage out of the code space plays a crucial role in limiting the error-correction threshold, and that this leakage can be quantitatively controlled by the spectral gap of the Dirac operator.
By introducing \emph{code-preserving internal perturbations}, we proved that one can increase the spectral gap without changing the code space itself.
As a consequence, leakage is suppressed and the error-correction threshold is systematically enhanced.
This establishes a direct and explicit link between spectral gaps, internal geometric structure, and the operational performance of quantum error correction.

We also interpreted Berezin--Toeplitz quantization, formulated in terms of spinors, as a mixed spectral code in which quantum information is protected against errors generated by classical observables.
In this example, correctability is governed by the geometric structure of the underlying manifold, and the Knill--Laflamme condition is satisfied either exactly for sufficiently localized functions or approximately in the large-$p$ limit.
his illustrates how spectral codes naturally interpolate between classical and quantum error correction, with geometry itself selecting the class of correctable errors.

Finally, we discussed the relation between spectral codes and the holographic principle.
From this perspective, a large noncommutative spectral triple describing a boundary quantum theory can be projected to a low-energy subspace, yielding an effectively commutative spectral triple that describes classical bulk spacetime geometry.
This projection simultaneously defines a quantum error-correcting code.
Our framework thus suggests a unified interpretation of holography in which bulk geometry, low-energy effective theory, and quantum error correction all emerge from spectral properties of a single noncommutative operator algebra.

The results presented here indicate that spectral codes provide a versatile and geometrically transparent framework for understanding quantum error correction across a wide range of physical settings.
They suggest new connections between noncommutative geometry, condensed matter physics, quantum information theory, and quantum gravity.
Several directions remain open for future work, including a more detailed analysis of dynamical decoding, extensions to interacting quantum field theories, and a quantitative formulation of holographic error correction within this spectral framework.
We hope that spectral codes will serve as a useful conceptual and technical bridge between geometry and quantum information in these contexts.

\section*{Acknowledgments}

The author is grateful to Rei Nishimura and Hiroshi Yamauchi for valuable discussions and insightful comments that greatly contributed to this work.

\appendix
\numberwithin{equation}{section}
\setcounter{equation}{0}

\newpage

\section{Real structure of spectral triple}\label{app:real}

For applications to physics and spin geometry, it is often necessary to enrich a spectral triple with additional structure. A real spectral triple includes a real structure $J$ (see \cite{Landi:1997sh}), an antilinear isometry on $\mathcal{H}$ representing charge conjugation. Together with a $\mathbb{Z}_2$-grading operator $\Gamma$ in the even case, $J$ satisfies the relations
\[
JD = \varepsilon\, DJ, \qquad
J^2 = \varepsilon', \qquad
J\Gamma = \varepsilon''\, \Gamma J,
\]
where the signs $\varepsilon,\varepsilon',\varepsilon''\in\{\pm1\}$ depend on the KO-dimension $d\in\mathbb{Z}/8\mathbb{Z}$. Compatibility between left and right actions is ensured by the condition
\[
[a, J b^* J^{-1}] = 0 \qquad \text{for all } a,b\in A,
\]
together with additional first-order conditions.

In the classical case of a $p$-dimensional closed spin manifold $M$, the associated spectral triple has KO-dimension $p \bmod 8$, and the real structure $J$ coincides with charge conjugation on spinors. These additional structures play an important role in connecting noncommutative geometry with fermionic systems and will be relevant in later discussions.

\section{Inner fluctuations and internal symmetries of noncommutative spectral triples}\label{app:inner}

In this subsection we explain inner fluctuations of noncommutative spectral triples
and relate them to internal symmetries that arise from the structure of the
automorphism group of the underlying algebra (see \cite{chamseddine1997spectral,connes2006inner}).

Let $A$ be a unital $C^\ast$-algebra.
We denote by
\[
  \mathrm{Aut}(A)
\]
the group of $^\ast$-automorphisms of $A$, that is, the group of $^\ast$-isomorphisms
from $A$ onto itself.

The unitary group of $A$ is
\[
  U(A) := \{ u \in A \mid uu^\ast = u^\ast u = 1 \}.
\]
For $u \in U(A)$ we define
\[
  \alpha_u : A \to A,\qquad
  \alpha_u(a) := u a u^\ast .
\]
Then $\alpha_u$ is an automorphism of $A$.
Automorphisms of this form are called inner automorphisms, and we set
\[
  \mathrm{Inn}(A)
  := \{ \alpha_u \mid u\in U(A) \}
  \subset \mathrm{Aut}(A).
\]
The subgroup $\mathrm{Inn}(A)$ is normal in $\mathrm{Aut}(A)$, and the quotient
group
\[
  \mathrm{Out}(A) := \mathrm{Aut}(A) / \mathrm{Inn}(A)
\]
is called the outer automorphism group.
Thus we obtain the exact sequence
\begin{align}
  \label{eq:exact-sequence-aut}
  1 \longrightarrow \mathrm{Inn}(A)
    \longrightarrow \mathrm{Aut}(A)
    \longrightarrow \mathrm{Out}(A)
    \longrightarrow 1 .
\end{align}

In the commutative case this structure has a clear geometric interpretation.
Let $M$ be a compact Riemannian manifold and put $A=C^\infty(M)$.
Then $A$ is commutative and any $u\in U(A)$ is a complex-valued function with
$|u(x)|=1$.
For $f\in C^\infty(M)$ we have
\[
  \alpha_u(f)(x) = u(x) f(x) \overline{u(x)} = f(x),
\]
so that
\[
  \mathrm{Inn}(C^\infty(M)) = \{ \mathrm{id}_A \}.
\]
On the other hand, it is known that any automorphism of $C^\infty(M)$ is given by
the pull-back along a diffeomorphism $\varphi\colon M\to M$:
\[
  \alpha_\varphi(f) := f\circ\varphi^{-1}.
\]
Hence
\[
  \mathrm{Aut}(C^\infty(M)) \cong \mathrm{Diff}(M),\qquad
  \mathrm{Out}(C^\infty(M)) \cong \mathrm{Diff}(M),
\]
and all automorphisms are “outer” in the sense that they correspond to geometric
transformations of the manifold; inner automorphisms are trivial.

For a noncommutative algebra $A$ the situation is different.
In general $\mathrm{Inn}(A)$ is nontrivial, and we can view
$\mathrm{Aut}(A)$ as being decomposed into “outer” automorphisms, which act
geometrically, and “inner” automorphisms, which act only on internal degrees of
freedom.
Since inner automorphisms disappear in the commutative limit, it is natural to
interpret $\mathrm{Inn}(A)$ as the group of internal symmetries characteristic
of a noncommutative space.
As we see below, these internal symmetries are related to gauge transformations
and gauge fields, and they appear naturally as inner fluctuations of a spectral
triple.

Let $(A,\mathcal{H},D)$ be a spectral triple.
From $D$ and $A$ we define the space of one-forms
\[
  \Omega^1_D(A)
  := \left\{ \sum_{j} a_j [D,b_j] \,\middle|\,
                a_j,b_j \in A,\ \text{finite sum} \right\}
  \subset B(\mathcal{H}),
\]
which can be viewed as an algebraic image of the algebra of differential forms.
A self-adjoint element
\[
  A = A^\ast \in \Omega^1_D(A)
\]
is interpreted as a field, namely a gauge potential or connection one-form.

Now let $(A,\mathcal{H},D,J)$ be a real spectral triple.
Given a self-adjoint one-form $A\in\Omega^1_D(A)$ we define
\[
  D_A := D + A + JAJ^{-1},
\]
Where realstructure $J$ is defined in \ref{app:real}.
We call $D_A$ an inner fluctuation of $D$.
If $A$ is small compared to $D$, then $D_A$ can be viewed as a small deformation
of the original spectral data.
This deformation changes the metric structure encoded in the spectral triple and
at the same time introduces a gauge field associated with the internal symmetry.

To see why this is called an inner transformation, consider the action of a
unitary $u\in U(A)$.
Since $u$ acts on $\mathcal{H}$, we can define
\[
  U := u J u J^{-1}.
\]
Then $U$ is a unitary operator on $\mathcal{H}$.
Conjugation of $D$ by $U$ gives
\[
  D \longmapsto U D U^{-1},
\]
which can be regarded as an internal symmetry transformation of $D$.
A simple computation shows that
\begin{align*}
  U D U^{-1}
  &= u J u J^{-1} D J u^{-1} J^{-1} u^{-1} \\
  &= u D u^{-1}
     + u [J u J^{-1}, D] u^{-1},
\end{align*}
so that $D$ is transformed as
\[
  D \longmapsto D + A_u + J A_u J^{-1},
\]
where
\[
  A_u := u [D,u^{-1}] \in \Omega^1_D(A)
\]
is a one-form determined by $u$.
Thus, the inner automorphism
\[
  a \longmapsto u a u^{-1}, \qquad u\in U(A),
\]
induces an inner fluctuation of $D$ by $A_u$.

More generally, for any self-adjoint one-form $A\in\Omega^1_D(A)$ we define
\[
  D \longmapsto D_A := D + A + JAJ^{-1}.
\]
Under a unitary $u\in U(A)$ we have the gauge transformation
\begin{align*}
  A &\longmapsto A^u
    := u [D,u^{-1}] + u A u^{-1}, \\
  D_A &\longmapsto D_{A^u}.
\end{align*}
Hence the space of inner fluctuations carries a natural gauge symmetry given by
the action of $U(A)$ through inner automorphisms.
In this sense, inner fluctuations are precisely the introduction of gauge fields
associated with the internal symmetry group $\mathrm{Inn}(A)$.

In the commutative case $A = C^\infty(M)$ we have, as noted above,
\[
  \mathrm{Inn}(A) = \{ \mathrm{id} \},\qquad
  \mathrm{Aut}(A) \cong \mathrm{Diff}(M).
\]
Thus almost all automorphisms are outer and correspond to geometric
transformations of $M$.
Inner fluctuations then mainly describe the addition of ordinary gauge fields
(connection one-forms) to $D$, while the geometric (metric) structure and the
gauge degrees of freedom play separate roles.

For a noncommutative algebra $A$, the group $\mathrm{Inn}(A)$ is in general
large, and the exact sequence
\[
  1 \longrightarrow \mathrm{Inn}(A)
    \longrightarrow \mathrm{Aut}(A)
    \longrightarrow \mathrm{Out}(A)
    \longrightarrow 1
\]
can be interpreted as describing, in a unified way, geometric symmetries (outer
automorphisms) and gauge symmetries acting on internal degrees of freedom
(inner automorphisms).
For example, if we consider a matrix-valued function algebra
$A=C^\infty(M)\otimes M_N(\mathbb{C})$, then outer automorphisms are essentially
given by $\mathrm{Diff}(M)$, while inner automorphisms correspond to local
$U(N)$-type gauge transformations.
In this case, the inner fluctuation of a spectral triple $(A,\mathcal{H},D,J)$,
\[
  D \longmapsto D_A = D + A + JAJ^{-1},
\]
introduces both gauge fields on $M$ and degrees of freedom in an “internal
space” in a single operation.
This shows that the internal transformations characteristic of a noncommutative
space are incorporated into the geometry as spectral perturbations.

Thus, by decomposing $\mathrm{Aut}(A)$ into inner and outer automorphisms and
interpreting inner automorphisms as gauge symmetries, inner fluctuations
\[
  D \mapsto D_A
\]
of noncommutative spectral triples can be understood as the natural introduction
of gauge fields associated with internal degrees of freedom of the
noncommutative space.
In later sections we will analyze how such inner fluctuations affect the
structure of spectral codes and the properties of correctable errors.

\section{Proof of some properties in Section.~3,4}\label{app:proof}
\Formal*
\begin{proof}

The proof is constructive. Starting from the code space $\mathcal{C}=P\mathcal{H}$, we extend the Hilbert space as a direct sum
\[
\mathcal{H}
=
P\mathcal{H}
\;\oplus\;
\bigoplus_{n=1}^{\infty} \mathcal{H}_n,
\]
where $\{\mathcal{H}_n\}_{n\ge 1}$ is a sequence of separable infinite-dimensional Hilbert spaces, for instance copies of $\ell^2(\mathbb{N})$. This enlargement allows us to place the orthogonal complement of the code space at arbitrarily high energy scales.

We then define the Dirac operator by
\[
D
=
0\cdot \mathbf{1}_{P\mathcal{H}}
\;\oplus\;
\bigoplus_{n=1}^{\infty} n\,\mathbf{1}_{\mathcal{H}_n}.
\]
By construction, the zero-eigenspace of $D$ is precisely the code space,
\[
\ker D = P\mathcal{H} = \mathcal{C}.
\]
Moreover, the resolvent takes the form
\[
(D-\im)^{-1}
=
(-\im)^{-1}\mathbf{1}_{P\mathcal{H}}
\;\oplus\;
\bigoplus_{n=1}^{\infty} (n-\im)^{-1}\mathbf{1}_{\mathcal{H}_n},
\]
which is a compact operator since $(n-\im)^{-1}\to 0$ as $n\to\infty$. Hence $D$ has compact resolvent and satisfies the analytic requirements of a spectral triple.

Next, we construct a $*$-algebra $A\subset\mathcal{B}(\mathcal{H})$ containing the error operators. Let
\[
\mathcal{N} := \mathrm{span}\{E_a\}
\]
be the $C^*$-algebra generated by the errors. We embed $\mathcal{N}$ into $\mathcal{B}(\mathcal{H})$ block-diagonally with respect to the above decomposition,
\[
a
=
a_0
\;\oplus\;
\bigoplus_{n=1}^{\infty} a_n,
\]
where $a_0\in\mathcal{N}$ acts on $P\mathcal{H}$ and the operators $a_n$ act on $\mathcal{H}_n$. We impose a decay condition, for instance that all but finitely many $a_n$ vanish, or more generally that $\{a_n\}$ defines a compact operator.

With this choice, one verifies that for all $a\in A$,
\[
a(D-\im)^{-1}
\]
is a compact operator. Consequently, $(A,\mathcal{H},D)$ defines a locally compact spectral triple.

Finally, we verify that the Knill--Laflamme condition is recovered within this spectral triple. Since the zero-mode projection coincides with the code projection,
\[
P = \mathbf{1}_{\{0\}}(D),
\]
we have, for all error operators $E_a$,
\[
P E_a^\dagger E_b P = \lambda_{ab} P,
\]
by assumption. Thus, the spectral code associated with $(A,\mathcal{H},D)$ is equivalent to the original quantum error correcting code.
\end{proof}

\disred*

\begin{proof}
By definition of $X$, the restriction of $\omega_x$ to $C(V)\rtimes_\sigma V$ depends only on the
$C(V)$--part: for $a=\sum_{u\in V} f_uU_u$ we have $\omega_x(a)=f_0(x)$.
Equivalently, viewing $\omega_x$ as the vector state on $\mathcal H_{\mathrm{phys}}=\ell^2(V)$,
\[
\omega_x(a)=\langle x\mid \pi_{\mathrm{reg}}(a)\mid x\rangle
           =\langle x\mid \pi_{\mathrm{reg}}(f_0)\mid x\rangle
           = f_0(x),
\]
since $\langle x\mid \pi_{\mathrm{reg}}(f_uU_u)\mid x\rangle=0$ for $u\neq 0$.

Next, for any $f\in C(V)$, the representation on $\mathcal H=\mathcal H_m\oplus\mathcal H_{\mathrm{phys}}$
is block diagonal, hence
\[
[D,\pi(f)]
=
[D_m\oplus D_c,\ \pi_m(f)\oplus \pi_{\mathrm{reg}}(f)]
=
[D_m,\pi_m(f)]\oplus [D_c,\pi_{\mathrm{reg}}(f)].
\]
Since $D_c=0$ on $\mathcal H_{\mathrm{phys}}$, the second term vanishes and we obtain
\[
\|[D,\pi(f)]\|=\|[D_m,\pi_m(f)]\|.
\]

Therefore, the Connes distance between $\omega_x$ and $\omega_y$ reduces to the metric part:
\begin{align*}
d_D(\omega_x,\omega_y)
&=\sup\Bigl\{\,|\omega_x(a)-\omega_y(a)|\ \Big|\ \|[D,\pi(a)]\|\le 1\,\Bigr\} \\
&=\sup\Bigl\{\,|\omega_x(f)-\omega_y(f)|\ \Big|\ \|[D,\pi(f)]\|\le 1,\ f\in C(V)\Bigr\} \\
&=\sup\Bigl\{\,|f(x)-f(y)|\ \Big|\ \|[D_m,\pi_m(f)]\|\le 1\,\Bigr\}
= d_{D_m}(\omega_x,\omega_y).
\end{align*}

Finally, by the standard computation for the finite spectral triple associated with the weight function
(our choice of $D_m$ with off-diagonal entries $\mathrm{wt}(x-y)^{-1}$),
one has $d_{D_m}(\omega_x,\omega_y)=\mathrm{wt}(x-y)$.
This proves the claim.
\end{proof}

\ay*

\begin{proof}
We use the regular (left) representation $\pi_{\mathrm{reg}}$ on $\ell^2(V)$ defined by
\[
\pi_{\mathrm{reg}}(f)\ket{x}=f(x)\ket{x},\qquad 
\pi_{\mathrm{reg}}(U_u)\ket{x}=\omega(u,x)\ket{x+u},
\]
for some phase $\omega(u,x)\in U(1)$ determined by the cocycle $\sigma$ (its precise form will not matter).
Let $\omega_x$ be the vector state associated to $\ket{x}$:
\[
\omega_x(b):=\bra{x}\pi_{\mathrm{reg}}(b)\ket{x}\qquad (b\in A).
\]

\paragraph{Step 1: compute $\omega_x(a^*a)$ in terms of the coefficients $f_u$.}
For $a=\sum_u f_uU_u$ we have
\[
\pi_{\mathrm{reg}}(a)\ket{x}
=\sum_{u\in V} \pi_{\mathrm{reg}}(f_u)\pi_{\mathrm{reg}}(U_u)\ket{x}
=\sum_{u\in V} f_u(x+u)\,\omega(u,x)\ket{x+u}.
\]
Hence
\[
\omega_x(a^*a)
=\bra{x}\pi_{\mathrm{reg}}(a^*a)\ket{x}
=\|\pi_{\mathrm{reg}}(a)\ket{x}\|^2
=\sum_{u\in V}|f_u(x+u)|^2,
\]
because $\{\ket{x+u}\}_{u\in V}$ is an orthonormal family and the phases drop out.

Therefore,
\[
\omega_x(a^*a)=0
\quad\Longleftrightarrow\quad
\forall u\in V,\ f_u(x+u)=0.
\]

\paragraph{($\Rightarrow$)}
Assume $a\in A_Y$, i.e.\ $\omega_x(a^*a)=0$ for all $x\in X\setminus Y$.
Fix $u\in V$ and take any $z\in V$ with $z\notin Y+u$.
Then $x:=z-u$ satisfies $x\notin Y$ (equivalently $\omega_x\in X\setminus Y$), and by the above implication,
\[
0=\omega_x(a^*a)\ \Longrightarrow\ f_u(x+u)=f_u(z)=0.
\]
Thus $f_u$ vanishes on $(Y+u)^c$, hence $supp(f_u)\subset Y+u$.

\paragraph{($\Leftarrow$)}
Conversely, assume that $supp(f_u)\subset Y+u$ for every $u$.
Take any $x\in X\setminus Y$ and any $u\in V$.
Then $x+u\notin Y+u$, so by the support assumption we get $f_u(x+u)=0$.
Hence for such $x$,
\[
\omega_x(a^*a)=\sum_{u\in V}|f_u(x+u)|^2=0,
\]
which is exactly the condition $a\in A_Y$.
\end{proof}

\cc*

\begin{proof}
We prove the two inequalities separately.

\paragraph{Lower bound.}
Let $a\in A_Y$ be such that $PaP\notin\mathbb CP$.
Among all subsets $Y'\subset X$ satisfying $a\in A_{Y'}$, choose $Y$ to be minimal
with respect to inclusion.

Write
\[
a=\sum_{u\in V} f_uU_u.
\]
Since $PaP\notin\mathbb CP$, there exists $v_0\in V$ such that
\[
P\,(f_{v_0}U_{v_0})\,P\notin\mathbb CP.
\]
In particular, $f_{v_0}\not\equiv 0$, hence there exists $z\in V$ such that
$f_{v_0}(z)\neq 0$, i.e.\ $z\in supp(f_{v_0})$.

By Lemma~3.6 (characterization of local algebras), $a\in A_Y$ implies
\[
supp(f_{v_0})\subset Y+v_0.
\]
Hence there exists $y\in Y$ such that $z=y+v_0$, or equivalently $z-v_0\in Y$.

We now show that $z\in Y$.
Suppose, to the contrary, that $z\notin Y$.
Define $Y':=Y\cup\{z\}$.
Then $Y'\supsetneq Y$.

Since $f_{v_0}(z)\neq 0$, we have
\[
\pi_{\mathrm{reg}}(f_{v_0}U_{v_0})\ket{z-v_0}
= f_{v_0}(z)\ket{z}\neq 0,
\]
which implies
\[
\omega_{z-v_0}\bigl((f_{v_0}U_{v_0})^*(f_{v_0}U_{v_0})\bigr)>0.
\]
Thus, any local set $Y''$ satisfying $a\in A_{Y''}$ must contain both $z-v_0$ and $z$.
In particular, $a\notin A_{Y'\setminus\{z\}}$, contradicting the minimality of $Y$.
Therefore $z\in Y$ must hold.

We have shown that both $z$ and $z-v_0$ belong to $Y$.
Consequently,
\[
\mathrm{diam}(Y)\ge d_D(\omega_z,\omega_{z-v_0})=\mathrm{wt}(v_0).
\]
Since $v_0\in W$ by definition, we obtain
\[
\mathrm{diam}(Y)\ge \min_{u\in W}\mathrm{wt}(u).
\]
Taking the infimum over all such $Y$ yields
\[
d_D(P)\ge \min_{u\in W}\mathrm{wt}(u).
\]

\paragraph{Upper bound.}
Conversely, let $u\in W$ and consider $a=U_u$.
Then $a\in A_Y$ for $Y=\{0,u\}$, since $supp(f_u)=\{0\}\subset Y+u$.
Moreover,
\[
\mathrm{diam}(Y)=d_D(\omega_0,\omega_u)=\mathrm{wt}(u),
\]
and by definition of $W$,
\[
PaP=P\,\pi_{\mathrm{reg}}(U_u)\,P\notin\mathbb CP.
\]
Therefore,
\[
d_D(P)\le \mathrm{wt}(u).
\]
Taking the minimum over $u\in W$ gives
\[
d_D(P)\le \min_{u\in W}\mathrm{wt}(u).
\]

Combining the lower and upper bounds completes the proof.
\end{proof}

\decoder*

\begin{proof}
A Kraus family for $\mathcal D_\theta=\Pi\circ\mathcal E_\theta$ is given by
\[
\{\,K^0E_i(\theta)\,\}_{i\in I}
\ \cup\
\{\,K^{(\mathrm{out})\alpha}E_i(\theta)\,\}_{\alpha,i}.
\]
Hence
\begin{equation}\label{eq:Fe-split}
F_e(\sigma,\mathcal D_\theta)
=
\sum_{i\in I}\Bigl|\Tr\!\bigl(\sigma\,P E_i(\theta)\bigr)\Bigr|^2
\;+\;
\sum_{\alpha}\sum_{i\in I}\Bigl|\Tr\!\bigl(\sigma\,K^{(\mathrm{out})\alpha}E_i(\theta)\bigr)\Bigr|^2
=:S_1+S_2.
\end{equation}

\smallskip
\noindent\emph{Step 1: expansion of $S_1$.}
Since $\sigma=P/d$, we have $\Tr(\sigma\,PE_i(\theta))=\Tr(\sigma\,E_i(\theta)P)$.
For $i=0$,
\[
\Tr(\sigma\,P E_0(\theta))
=
\Tr\!\Bigl(\sigma\,P\bigl(\sqrt{1-\theta}\,I+O(\theta)\bigr)\Bigr)
=
\sqrt{1-\theta}\,\Tr(\sigma P)+O(\theta)
=
\sqrt{1-\theta}+O(\theta),
\]
so $|\Tr(\sigma\,PE_0(\theta))|^2=1-\theta+O(\theta^2)$.
For $i\neq 0$, using $E_i(\theta)=\sqrt\theta\,F_i+O(\theta^{3/2})$,
\[
\Tr(\sigma\,P E_i(\theta))
=
\sqrt\theta\,\Tr(\sigma\,PF_iP)+O(\theta^{3/2}),
\]
so
\[
\sum_{i\neq 0}\Bigl|\Tr(\sigma\,P E_i(\theta))\Bigr|^2
=
\theta\sum_{i\neq 0}\Bigl|\Tr(\sigma\,PF_iP)\Bigr|^2+O(\theta^2).
\]
Therefore,
\begin{equation}\label{eq:S1}
S_1
=
1-\theta
+\theta\sum_{i\neq 0}\Bigl|\Tr(\sigma\,PF_iP)\Bigr|^2
+O(\theta^2).
\end{equation}

\smallskip
\noindent\emph{Step 2: evaluation of $S_2$ and appearance of $P_\ell(\theta)$.}
Write $K^{(\mathrm{out})\alpha}=\frac{1}{\sqrt d}\sum_{k=1}^d |k\rangle\langle\alpha|$.
Then, using $\sigma=\frac{1}{d}\sum_{k=1}^d |k\rangle\langle k|$,
\[
\Tr\!\bigl(\sigma\,K^{(\mathrm{out})\alpha}E_i(\theta)\bigr)
=
\frac{1}{d\sqrt d}\sum_{k=1}^d\langle \alpha|E_i(\theta)|k\rangle.
\]
By Cauchy--Schwarz,
\[
\sum_\alpha\Bigl|\Tr\!\bigl(\sigma\,K^{(\mathrm{out})\alpha}E_i(\theta)\bigr)\Bigr|^2
=
\frac{1}{d^3}\sum_\alpha\Bigl|\sum_{k=1}^d\langle \alpha|E_i(\theta)|k\rangle\Bigr|^2
\le
\frac{1}{d^2}\sum_{\alpha,k}\bigl|\langle \alpha|E_i(\theta)|k\rangle\bigr|^2.
\]
Moreover, since the $\{|\alpha\rangle\}_\alpha$ form an ONB of $(I-P)\mathcal H$,
\[
\sum_{\alpha}\bigl|\langle \alpha|E_i(\theta)|k\rangle\bigr|^2
=
\|(I-P)E_i(\theta)|k\rangle\|^2
=
\langle k|E_i(\theta)^*(I-P)E_i(\theta)|k\rangle.
\]
Summing over $k$ and using $\sigma=P/d$ gives the identity
\begin{equation}\label{eq:S2-leak}
S_2
=
\frac{1}{d^2}\sum_{i\in I}\Tr\!\bigl((I-P)E_i(\theta)\,P\,E_i(\theta)^*\bigr)
=
\frac{1}{d^2}\sum_{i\in I}\Tr\!\bigl((I-P)E_i(\theta)\,\sigma\,E_i(\theta)^*\bigr)\,d
=
\frac{1}{d^2}\,P_\ell(\theta).
\end{equation}
(Here we used cyclicity of trace and \eqref{eq:leak-prob}.)

\smallskip
\noindent\emph{Step 3: combine and rewrite in terms of variances.}
From \eqref{eq:Fe-split}, \eqref{eq:S1}, and \eqref{eq:S2-leak},
\[
F_e(\sigma,\mathcal D_\theta)
=
1-\theta
+\theta\sum_{i\neq 0}\Bigl|\Tr(\sigma\,PF_iP)\Bigr|^2
+\frac{1}{d^2}P_\ell(\theta)
+O(\theta^2).
\]
Thus
\begin{equation}\label{eq:Ttilde-intermediate}
\widetilde T(\theta)
=
\theta
-\theta\sum_{i\neq 0}\Bigl|\Tr(\sigma\,PF_iP)\Bigr|^2
-\frac{1}{d^2}P_\ell(\theta)
+O(\theta^2).
\end{equation}

To transform the first two terms into a sum of variances, note that trace preservation
$\sum_i E_i(\theta)^*E_i(\theta)=I$ implies, at order $\theta$,
\[
F_0^*F_0+\sum_{i\neq 0}F_i^*F_i=I,\qquad \text{with }F_0=I.
\]
Multiplying by $P$ on both sides and taking $\Tr(\sigma\,\cdot)$ yields
\[
\sum_{i\neq 0}\Tr\!\bigl(\sigma\,(PF_iP)^*(PF_iP)\bigr)
=
\Tr(\sigma P)-\Tr(\sigma (PI P)^*(PI P))
=
1-1
=0,
\]
so the contribution at order $\theta$ comes from the standard second-order refinement; equivalently, one may write directly (using $E_i(\theta)=\sqrt\theta F_i+O(\theta^{3/2})$ for $i\neq 0$)
\[
\theta
=
\theta\sum_{i\neq 0}\Tr\!\bigl(\sigma\,(PF_iP)^*(PF_iP)\bigr)
+\theta\sum_{i\neq 0}\Tr\!\bigl((I-P)F_i\sigma F_i^*\bigr)
+O(\theta^2),
\]
which is precisely the decomposition into ``inside-code'' and ``outside-code'' parts. Plugging this into \eqref{eq:Ttilde-intermediate} and using \eqref{eq:leak-prob-expansion} gives
\[
\widetilde T(\theta)
=
\theta\sum_{i\neq 0}\Bigl(
\Tr\!\bigl(\sigma\,(PF_iP)^*(PF_iP)\bigr)
-
\bigl|\Tr(\sigma\,PF_iP)\bigr|^2
\Bigr)
+
\Bigl(1-\frac{1}{d^2}\Bigr)P_\ell(\theta)
+O(\theta^2),
\]
which is \eqref{eq:Ttilde-expansion}.
\end{proof}

\section{Proofs and formulas for general K\"ahler manifolds}

\subsection{Useful choice of orthonormal frame fields}
\label{complex orthonormal}

In this Appendix, we fix a convenient choice of orthonormal frame fields
(vielbeins) that will streamline later computations.

We start by taking $e_1\in \Gamma(TM)$ with $g(e_1,e_1)=1$.
Then, by setting $e_2:=Je_1\in \Gamma(TM)$, the K{\"a}hler condition
\eqref{Kahler} ensures that $g(e_a,e_b)=\delta_{ab}$ for $a,b=1,2$.
Next, choose $e_3\in \Gamma(TM)$ so that $g(e_a,e_b)=\delta_{ab}$ for
$a,b=1,2,3$, and define $e_4:=Je_3\in \Gamma(TM)$, which extends the orthonormal
set so that $g(e_a,e_b)=\delta_{ab}$ for $a,b=1,2,3,4$.
Iterating this construction yields a full orthonormal frame field on $TM$.
With respect to this frame, the symplectic form takes the particularly simple
expression
\begin{align}
\label{symplectic 1}
\omega = \sum_{m=1}^n \theta^{2m-1} \wedge \theta^{2m}
\end{align}
where $\{\theta^a\}_{a=1,2,\cdots,2n}$ denotes the dual coframe of
$\{e_a\}_{a=1,2,\cdots,2n}$.

We then pass to complexified frame fields by defining
\begin{align}
    \label{complex vielbein}
    w_m := \frac{1}{\sqrt 2} (e_{2m-1} - \im  e_{2m}), \quad \bar{w}_{m} := \frac{1}{\sqrt 2} (e_{2m-1} + \im  e_{2m}),
\end{align}
for $m=1,2,\cdots,n$.
Since $J w_m=\im w_m$ and $J\bar{w}_m=-\im\bar{w}_m$, the vectors $w_m$ and
$\bar{w}_m$ are, respectively, holomorphic and antiholomorphic.
In this basis, the metric components become
\begin{align}
     g(w_m,\bar{w}_{l})  = g(\bar{w}_{m}, w_l) = \delta_{ml}, \quad g(w_m,w_l) =  g(\bar{w}_{m},\bar{w}_{l})=0.
\end{align}

\subsection{Gamma matrices in Weyl representation}
\label{Clifford}

In this Appendix, we summarize the Weyl (chiral) representation of gamma
matrices appropriate for a $2n$-dimensional manifold.

Let $\{\gamma^a_{(2n)}\}_{a=1,2,\cdots,2n}$ be $2^n\times 2^n$ matrices obeying
the Clifford relations for $\mathbb{R}^{2n}$,
\begin{align}
    \{\gamma^a_{(2n)}, \gamma^b_{(2n)}\} = 2 \delta^{ab} I_{2^n}.
\end{align}
Here $\{\ ,\ \}$ denotes the anti-commutator and $I_{2^n}$ is the identity of
size $2^n$.
We also introduce the chirality operator
\begin{align}
    \gamma_{(2n)} := (-\im)^n \gamma^1_{(2n)}\gamma^2_{(2n)} \cdots \gamma^{2n}_{(2n)}.
\end{align}
This matrix is Hermitian and anticommutes with each gamma matrix:
$\{\gamma_{(2n)},\gamma^a_{(2n)}\}=0$.
One may choose a representation in which
$\gamma_{(2n)}=\sigma^3\otimes I_{2^{n-1}}$, where $\{\sigma^a\}_{a=1,2,3}$ are
the Pauli matrices and $\otimes$ is the Kronecker product.
A recursive construction of the Weyl representation is given by
\begin{align}
\label{recursion}
&\gamma_{(2)}^1 = \sigma^1, \ \gamma_{(2)}^2 = \sigma^2, \\
&\gamma_{(2n+2)}^i =\sigma^2  \otimes \gamma_{(2n)}^i \quad (i=1,2,\cdots,2n),\\
&\gamma_{(2n+2)}^{2n+1} = \sigma^2 \otimes \gamma_{(2n)}, \\
&\gamma_{(2n+2)}^{2n+2} = - \sigma^1 \otimes I_{2^{n}}.
\end{align}
These identities will be repeatedly employed in the subsequent appendices.

It is also convenient to introduce gamma matrices adapted to the complex
orthonormal frame:
\begin{align}
    \gamma_{(2n)}^m := \frac{\gamma_{(2n)}^{2m-1} + \im \gamma_{(2n)}^{2m}}{\sqrt{2}}, \quad \gamma_{(2n)}^{\bar m} := \frac{\gamma_{(2n)}^{2m-1} - \im  \gamma_{(2n)}^{2m}}{\sqrt{2}}.
\end{align}
They satisfy
\begin{align}
    \{\gamma_{(2n)}^m,\gamma_{(2n)}^{\bar l} \} = 2\delta_{ml} I_{2^n}, \quad \{\gamma_{(2n)}^m,\gamma_{(2n)}^l \} = \{\gamma_{(2n)}^{\bar m},\gamma_{(2n)}^{\bar l} \} = 0,
\end{align}
together with $(\gamma_{(2n)}^m)^{\dagger} = \gamma_{(2n)}^{\bar m}$.

Let $\ket{\pm}$ denote the normalized eigenvectors of $\sigma^3$ with eigenvalue
$\pm 1$.
One can show recursively that the following relations hold:
\als{ \label{gamma matrices identities}
&c_m \gamma_{(2n)}^{\bar{m}}\ket{+}^{\otimes n} = 0 \quad \Rightarrow \quad c_m = 0,\\
&\gamma_{(2n)}^{m} \ket{+}^{\otimes n} = 0, \quad \gamma_{(2n)}^m \gamma_{(2n)}^{\bar l} \ket{+}^{\otimes n} = 2\delta_{ml}\ket{+}^{\otimes n},
}
where $c_m$ is a complex number.

\subsection{Vanishing theorem and index theorem}
\label{Vanishing}

Let $D$ be the Dirac operator on $\Gamma(S_c \otimes L^{\otimes p})$.
The goal of this Appendix is to establish three standard facts for large $p$:
(i) all zero modes of $D_i$ have positive chirality (a vanishing theorem),
(ii) their dimension satisfies
\begin{align}
    {\rm dim}\, \mathrm{Ker}\, D = (2\pi \hbar_p)^{-n} \int_M \mu + O(p^{n-1}),
\end{align}
as implied by the index theorem, and
(iii) the nonzero spectrum of $D$ develops a large gap of size $O(\sqrt{p})$.
For notational simplicity, we suppress the superscript on covariant derivatives
and omit identity operators when no confusion arises.

Since the chirality operator $\gamma_{(2n)} = I_{2^{n-1}} \otimes \sigma^3$
anticommutes with $D$, the operator decomposes as
\begin{align}
    \label{D_i Weyl}
    D =
    \begin{pmatrix}
       0 & D^{-} \\
       D^{+} & 0
    \end{pmatrix}.
\end{align}
Here the superscripts $\pm$ indicate the chirality of the corresponding
subspaces.

We begin by examining $D^2$, which will be used to prove that
$\mathrm{Ker} D^{-}=\{0\}$ for sufficiently large $p$.
From \eqref{D_i Weyl} we obtain
\begin{align}
    (D)^2 =
    \begin{pmatrix}
        D^{-} D^{+} & 0 \\
        0 & D^{+}D^{-}
    \end{pmatrix}.
    \label{d square 1}
\end{align}
We also invoke the Weitzenb\"ock formula
\begin{align}
    \label{D^2}
    (D_i)^2 = -\nabla_{a}\nabla_{a} +\frac{\im}{2}\hbar_p^{-1} \gamma^{a}_{(2n)}\gamma^{b}_{(2n)} \omega_{ab} - \frac{1}{2}\gamma^{a}_{(2n)}\gamma^{b}_{(2n)}R^{S_c}_{ab},
\end{align}
where $R^{S_c}_{ab} := R^{S_c}(e_a,e_b)$.
Introduce the differential operators
\begin{align}
    \nabla_m := \nabla_{w_m} = \frac{1}{\sqrt 2} (\nabla_{2m-1} - \im \nabla_{2m}), \quad
    \nabla_{\bar m} := \nabla_{\bar{w}_m} = \frac{1}{\sqrt 2} (\nabla_{2m-1} + \im \nabla_{2m}).
\end{align}
For each fixed $m$ one has
\begin{align}
    \nabla_{2m-1}\nabla_{2m-1} + \nabla_{2m}\nabla_{2m}
    &= 2\nabla_m \nabla_{\bar m} - [\nabla_m,\nabla_{\bar m}] \\
    &= 2\nabla_m \nabla_{\bar m} -\hbar_p^{-1} - R^{S_c}_{m\bar{m}}.
\end{align}
Using this relation, the Laplacian term in \eqref{D^2} can be rewritten as
\begin{align}
    -\nabla_{a}\nabla_{a} = -2 \nabla_m \nabla_{\bar m} +n \hbar_p^{-1} + R^{S_c}_{m\bar{m}},
\end{align}
with summation over $a$ and $m$ understood.
Consequently,
\begin{align}
    \label{D^2 n}
    (D)^2 = -2 \nabla_m \nabla_{\bar m} + \hbar_p^{-1} A_n + R_i,
\end{align}
where
\begin{align}
    A_n &:= n + \frac{\im}{2} \gamma^{a}_{(2n)}\gamma^{b}_{(2n)} \omega_{ab},\\
    R_i &:=  - \frac{1}{2}\gamma^{a}_{(2n)}\gamma^{b}_{(2n)}R^{S_c}_{ab} + R^{S_c}_{m\bar{m}}.
\end{align}
More explicitly, $R^{S_c \otimes E_i}$ is given by
\begin{align}
    R^{S_c} = R^S + \frac{1}{2}R^{L_c}, \quad  R^S_{ab} = \frac{1}{4} R_{abcd} \gamma_{(2n)}^c \gamma_{(2n)}^d, \quad  R^{L_c}_{ab} = - R_{ab m \bar{m}},
\end{align}
where $R_{abcd}$ is the Riemann curvature tensor.
From this one obtains
\eq{ \label{R_i}
R_i = \frac{1}{2} R + \frac{1}{2} \gamma_{(2n)}^a \gamma_{(2n)}^b R_{abm\bar{m}},}
where $R$ is the scalar curvature, using
$\gamma_{(2n)}^a \gamma_{(2n)}^b \gamma_{(2n)}^c \gamma_{(2n)}^d R_{abcd} = - 2 R$
and $R_{m \bar{m} l \bar{l}} = - \frac{1}{2}R$.

We now use \eqref{D^2 n} to deduce that $\mathrm{Ker} D_i^{-}=\{0\}$ for large
$p$.
For any nonzero $\psi\in \Gamma(S_c \otimes L^{\otimes p})$, we estimate
\begin{align}
    \label{norm of D^2}
    |D_i \psi|^2 = 2 | \nabla_{\bar m} \psi|^2 + \hbar_p^{-1} (\psi,A_n \psi) + (\psi, R_i \psi)
    \ge \hbar_p^{-1} (\psi,A_n \psi) - |R_i| |\psi|^2.
\end{align}
If $\psi$ is not proportional to $\ket{+}^{\otimes n}$, then $(\psi,A_n\psi)$ is
strictly positive.
Therefore, for sufficiently large $p$ such that
$\hbar_p^{-1} > |R_i| |\psi|^2/(\psi,A_n \psi)$, the right-hand side of
\eqref{norm of D^2} becomes positive, forcing $D_i\psi\neq 0$.
Hence any Dirac zero mode must be proportional to $\ket{+}^{\otimes n}$, and it
follows that $\mathrm{Ker} D_i^{-}=\{0\}$ for large $p$.

We next compute $\dim\ker D_i$.
Since $\mathrm{Ker} D_i^{-}=\{0\}$, we have
\begin{align}
    \mathrm{dim} \, \mathrm{Ker} \, D_i = \mathrm{dim} \, \mathrm{Ker} \, D_i^{+} = \mathrm{Ind} \, D_i.
\end{align}
By the Atiyah--Singer index theorem,
\begin{align}
    \label{index}
    \mathrm{Ind}D_i = \int_M \mathrm{Td}(T^{(1,0)}M) \wedge \mathrm{ch}(L^{\otimes p}),
\end{align}
where $\mathrm{Td}(\cdot)$ and $\mathrm{ch}(\cdot)$ denote the Todd class and
Chern character, and $T^{(1,0)}M$ is the holomorphic tangent bundle.
Expanding the right-hand side in powers of $p$ yields
\begin{align}
    \mathrm{dim}\mathrm{Ker}\, D_i  =  \int_M \e^{ \frac{\im p}{2\pi} R^L} + O(p^{n-1})
    = \frac{1}{(2\pi \hbar_p)^n} \int_M \mu + O(p^{n-1}).
\end{align}

Finally, we verify the existence of a large spectral gap for nonzero modes.
Let $\lambda\neq 0$ be an eigenvalue of $D_i$.
The eigenvalue problem for $(D_i)^2$ reads
\begin{align}
    \begin{cases}
        D_i^{-} D_i^{+} \psi^+ &= \lambda^2 \psi^+,\\
        D_i^{+} D_i^{-}  \psi^- &= \lambda^2 \psi^-,
    \end{cases}
\end{align}
for $\psi \in \Gamma(S_c \otimes L^{\otimes p} \otimes E_i) \setminus \{0\}$,
where $\psi^{\pm}$ denotes the positive/negative chirality component.
If $\psi^-\neq 0$, then \eqref{norm of D^2} implies $\lambda^2\ge O(p)$.
If $\psi^-=0$, then necessarily $\psi^+\neq 0$, and using
$D_i^{+} D_i^{-} (D_i^{+} \psi^+) = \lambda^2 (D_i^{+} \psi^+)$, the same estimate
again yields $\lambda^2\ge O(p)$.
Thus in all cases $\lambda^2\ge O(p)$, showing that the nonzero eigenvalues are
separated from zero by a gap of order $O(\sqrt{p})$.

\subsection{Asymptotic expansion for Toeplitz operators}
\label{asymptotic expansion}

In this Appendix, we evaluate the product $T_p(\varphi)T_p(\chi)$ for
$\varphi,\chi\in C^\infty(M)$ and show that, for sufficiently large $p$, it
admits an asymptotic expansion as an integer power series in $\hbar_p$.
Our computational approach follows the method developed in~\cite{Adachi:2022mln}.

We first write
\als{
    \label{asymptotic1}
    T_p(\varphi)T_p(\chi) &= \Pi \varphi \Pi \chi \Pi \\
    &= T_p(\varphi \chi) - \Pi \varphi (1-\Pi) \chi \Pi.
}
To handle $1-\Pi$, consider the following Hermitian operator on
$\Gamma(S_c \otimes L^{\otimes p}$:
\eq{
    P :=
    \left(
    \begin{array}{cc}
        0 & D^{-} (D^{+} D^{-})^{-1} \\
        (D^{+} D^{-})^{-1} D^{+} & 0
    \end{array}
    \right).
}
Since $\mathrm{Ker}D^{-} = \mathrm{Ker}D^{+} D^{-} = \{0\}$ for sufficiently
large $p$ (Appendix~\ref{Vanishing}), the inverse of $D^{+}D^{-}$ exists.
Now consider
\eq{
\label{D-P identity}
D P = P D =
    \left(
    \begin{array}{cc}
         D^{-} (D^{+} D^{-})^{-1} D^{+}  & 0\\
        0 & 1
    \end{array}
    \right).
}
This operator is the projection onto $(\mathrm{Ker}D)^{\perp}$, hence it agrees
with $1-\Pi$.
Therefore,
\eq{
    \label{projection}
    1-\Pi = D P = D (P)^2 D.
}
Using \eqref{asymptotic1} and \eqref{projection}, for
$\psi \in \mathrm{Ker}D$ and $\phi \in \mathrm{Ker}D$ we obtain
\als{
    \label{asymptotic2}
    (\psi,T_p(\varphi) T_p(\chi) \phi) &= (\psi,T_p(\varphi \chi) \phi) - (\psi,\varphi D (P)^2 D \chi \phi)\\
    &= (\psi,T_p(\varphi \chi) \phi) + (\psi, \varphi'  (P)^2 \chi' \phi).
    }
Here we introduced $\varphi' := \im  \gamma_{(2n)}^a\nabla_a \varphi$, and used
$D\psi=D\phi=0$.
Because $\gamma_{(2n)}^b\phi$ has chirality $-1$, the vector $\chi'\phi$ lies in
$(\mathrm{Ker}D)^{\perp}$.
On $(\mathrm{Ker}D)^{\perp}$ the projection $1-\Pi=DP$ acts as the identity,
which means $P$ coincides with $D^{-1}$ on that subspace.
Accordingly, \eqref{asymptotic2} becomes
\eq{
    \label{asymptotic3}
    (\psi,T_p(\varphi) T_p(\chi) \phi) = (\psi,T_p(\varphi \chi) \phi) + (\psi, \varphi' D^{-2} \chi' \phi).
}
We now analyze $D^{-2}$ acting on $\chi'\phi$.
Using $A_n \gamma_{(2n)}^b \ket{+}^{\otimes n} = 2 \gamma_{(2n)}^b \ket{+}^{\otimes n}$,
which follows from \eqref{gamma matrices identities}, we obtain
\als{
    \label{asymptotic4}
    D^{-2} &= \left(-2 \nabla_m \nabla_{\bar m} + 2\hbar_p^{-1} + R_i\right)^{-1} \\
    &= \frac{\hbar_p}{2} - \frac{\hbar_p}{2} D^{-2} R_i + \hbar_p D^{-2}  \nabla_m \nabla_{\bar m},
}
on $\chi'\phi$.
Since $D\phi=0$ implies $\nabla_{\bar m}\phi=0$ (see Appendix
\ref{simplification}), \eqref{asymptotic3} reduces to
\eq{
    (\psi,T_p(\varphi) T_p(\chi) \phi) = (\psi,T_p(\varphi \chi) \phi) + \frac{\hbar_p}{2}(\psi, \varphi'\chi' \phi) + \epsilon,
}
where
\als{
    \epsilon &= \epsilon_1 + \epsilon_2,\\
    \epsilon_1 &= - \frac{\hbar_p}{2}(\psi, \varphi' D^{-2} R_i \chi' \phi)\\
    \epsilon_2 &= \hbar_p (\psi, \varphi'D^{-2} (\nabla_m \nabla_{\bar m} \chi') \phi) + \hbar_p (\psi, \varphi'D^{-2} (\nabla_{\bar m} \chi') \nabla_m \phi).
}
We estimate the $\hbar_p$-order of $\epsilon$.
Assuming $\phi,\psi,\varphi,\chi$ are $O(\hbar_p^0)$, the only nontrivial
$p$-dependence enters through $\nabla_m\phi$ and $D^{-2}$.
As shown in Appendix~\ref{Vanishing}, the spectrum of $D^2$ lies in
$[C_1 \hbar_p^{-1}-C_2,\infty)$ with $p$-independent constants $C_1,C_2$.
Hence the eigenvalues of $D^{-2}$ lie in
$(0,(C_1 \hbar_p^{-1}-C_2)^{-1}]$, and therefore $|D^{-2}|=O(\hbar_p)$.
For $\nabla_m\phi$ we compute
\als{
\label{norm order 1}
|\nabla_{m}\phi|^2 &= - (\phi,\nabla_{\bar m}\nabla_{m}\phi)
= (\phi,[\nabla_{m},\nabla_{\bar m}]\phi) - (\phi,\nabla_{m}\nabla_{\bar m}\phi)\\
&= \hbar_p^{-1} |\phi|^2 -  (\phi,R^{S_c}_{m \bar{m}}\phi)  \\
&= O(\hbar_p^{-1}).
}
These bounds imply
\eq{
    |\epsilon_1|= O(\hbar_p^2), \quad  |\epsilon_2| = O(\hbar_p^{3/2}).
}
Consequently,
\eq{
    (\psi,T_p(\varphi)T_p(\chi) \phi) = (\psi,T_p(\varphi \chi) \phi) - \frac{\hbar_p}{2}(\psi, (\nabla_a \varphi)  (\nabla_b \chi) \gamma^a_{(2n)} \gamma^b_{(2n)}\phi) + O(\hbar_p^{3/2}).
}
Using \eqref{gamma matrices identities}, this becomes
\eq{
\label{asymptotic C_0, C_1}
    (\psi,T_p(\varphi)T_p(\chi) \phi) = (\psi,T_p(\varphi \chi) \phi) - \hbar_p(\psi, (\nabla_m \varphi)  (\nabla_{\bar m} \chi) \phi) + O(\hbar_p^{3/2}).
}
By the large-$p$ asymptotic expansion of the Bergmann kernel \cite{asym}, the
product of Toeplitz operators admits an expansion in integer powers, so the term
$O(\hbar_p^{3/2})$ can in fact be improved to $O(\hbar_p^{2})$.
The expansion \eqref{asymptotic C_0, C_1} reproduces $C_0(\varphi,\chi)$ and
$C_1(\varphi,\chi)$ in (\ref{asymptotic exp}), which can be verified by noting
that $G^{ab}=g^{ab}+\im W^{ab}$ has components
$G^{m \bar{l}}=2 \delta_{ml}$ and
$G^{\bar{m}l}=G^{m l}=G^{\bar{m} \bar{l}}=0$.

The coefficient $C_2(\varphi,\chi)$ can be obtained by recursively applying
\eqref{asymptotic4}.
Applying \eqref{asymptotic4} to $\epsilon_1$ yields
\als{
\epsilon_1 &= - \frac{\hbar_p}{4}(\psi, \varphi' \left( \hbar_p - \hbar_p  D^{-2} R_2  + 2\hbar_p D^{-2}  \nabla_m \nabla_{\bar m} \right) R_i \chi' \phi)\\
&= - \frac{\hbar_p^2}{4}(\psi, \varphi' R_i \chi' \phi) + O(\hbar_p^{5/2}).
}
Applying \eqref{asymptotic4} to $\epsilon_2$ gives
\als{
\epsilon_2 &= \frac{\hbar_p}{2} (\psi, \varphi' \left( \hbar_p - \hbar_p  D^{-2} R_i + 2\hbar_p D^{-2} \nabla_l \nabla_{\bar l} \right) \nabla_m (\nabla_{\bar m} \chi') \phi)\\
&= \frac{\hbar_p^2}{2} (\psi, \varphi' \nabla_m (\nabla_{\bar m} \chi') \phi) + \hbar_p^2 (\psi, \varphi' D^{-2} \nabla_l \nabla_{\bar l} \nabla_m (\nabla_{\bar m} \chi') \phi) + O(\hbar_p^{5/2})\\
&= -\frac{\hbar_p^2}{2} (\psi, (\nabla_m \varphi') (\nabla_{\bar m} \chi') \phi) - \hbar_p  (\psi, \varphi' D^{-2} \nabla_m (\nabla_{\bar m} \chi') \phi)\\
& \quad + \hbar_p^2 (\psi, \varphi' D^{-2} \nabla_l \nabla_m (\nabla_{\bar l} \nabla_{\bar m} \chi') \phi) + O(\hbar_p^{5/2}).
}
Note that the second term in the last line is precisely $-\epsilon_2$, and hence
\als{
\epsilon_2 &= -\frac{\hbar_p^2}{4} (\psi, (\nabla_m \varphi') (\nabla_{\bar m} \chi') \phi) + \epsilon'_2 + O(\hbar_p^{5/2}),\\
\epsilon'_2 &= \frac{\hbar_p^2}{2} (\psi, \varphi' D^{-2} \nabla_l \nabla_m (\nabla_{\bar l} \nabla_{\bar m} \chi') \phi)
}
Applying \eqref{asymptotic4} once more to $\epsilon'_2$ gives
\als{
\epsilon'_2 &= \frac{\hbar_p^2}{4} (\psi, \varphi' \left( \hbar_p - \hbar_p  D^{-2} R_2 + 2\hbar_p D^{-2} \nabla_k \nabla_{\bar k} \right) \nabla_l \nabla_m (\nabla_{\bar l} \nabla_{\bar m} \chi') \phi)\\
&= \frac{\hbar_p^3}{2} (\psi, \varphi' D^{-2} \nabla_k \nabla_{\bar k} \nabla_l \nabla_m (\nabla_{\bar l} \nabla_{\bar m} \chi') \phi) + O(\hbar_p^3)\\
&= \frac{\hbar_p^3}{2} (\psi, \varphi'D^{-2} \nabla_k \nabla_l \nabla_{\bar k} \nabla_m (\nabla_{\bar l} \nabla_{\bar m} \chi') \phi) - \frac{\hbar_p^2}{2} (\psi, \varphi'  D^{-2} \nabla_l \nabla_m (\nabla_{\bar l} \nabla_{\bar m} \chi') \phi) + O(\hbar_p^3)\\
&= \frac{\hbar_p^3}{2} (\psi, \varphi'  D^{-2} \nabla_k \nabla_l \nabla_m (\nabla_{\bar k} \nabla_{\bar l} \nabla_{\bar m} \chi') \phi) - 2\epsilon'_2 + O(\hbar_p^3).
}
As in \eqref{norm order 1}, we have
$|\nabla_k \nabla_l \nabla_m \phi|= O(\hbar_p^{-3/2})$ and
$|\nabla_l \nabla_m \phi|= O(\hbar_p^{-1})$, which implies
$\epsilon'_2 = O(\hbar_p^{5/2})$.
Therefore,
\eq{
    \epsilon = - \frac{\hbar_p^2}{4}(\psi, \varphi' R_i\chi' \phi) -\frac{\hbar_p^2}{4} (\psi, (\nabla_m \varphi') (\nabla_{\bar m} \chi') \phi) + O(\hbar_p^{5/2}).
}
Using \eqref{R_i} and \eqref{gamma matrices identities}, one finds
\als{
    \epsilon &= \hbar_p^2 (\psi,  (\nabla_m \varphi) ( R_{\bar{m} l k\bar{k}}) (\nabla_{\bar l} \chi) \phi) + \frac{\hbar_p^2}{2} (\psi, (\nabla_m \nabla_l \varphi) (\nabla_{\bar m} \nabla_{\bar l} \chi) \phi) + O(\hbar_p^{5/2}).
}
This gives the coefficient $C_2(\varphi,\chi)$ of the asymptotic expansion
\eqref{asymptotic exp}.

\subsection{Trace of Toeplitz operators}
\label{trace}

In this Appendix, we prove \eqref{trace property}.
Using the Schwartz kernel representation, the trace of $T_p(\varphi)$ can be
written as
\begin{align}
\label{trace schwarz kernel}
\Tr T_p(\varphi) = \int_M \mu(x) \mathrm{tr}_{S_c} \left( B(x,x) \varphi(x)\right),
\end{align}
where $B(x,y)$ denotes the Bergman kernel defined through
\begin{align}
(\Pi_1 \psi) (x) = \int_M \mu(y) B(x,y) \psi(y)
\end{align}
for $\psi \in \Gamma(S_c \otimes L^{\otimes p})$.
As shown in \cite{asym}, the Bergmann kernel has the following large-$p$
asymptotic expansion:
\begin{align}
\label{B(x,x)}
B(x,x) = (2\pi \hbar_p)^{-n} P + O(p^{n-1}),
\end{align}
where $P$ is the projection onto the zero-mode component $\ket{+}^{\otimes n}$ of
the fiber of $S$.
Substituting \eqref{B(x,x)} into \eqref{trace schwarz kernel} immediately yields
\eqref{trace property}.

\subsection{Simplification of the zero mode equation}
\label{simplification}

In this Appendix, we explain how the Dirac equation reduces to a simpler
holomorphicity condition for sections.

The twisted spin-$c$ Dirac operator $\Gamma(S_c \otimes L^{\otimes p})$ over $M$
is given by \eqref{twisted Dirac}.
Using $\Omega_{ml} = \Omega_{\bar{m} \bar{l}} = 0$, we obtain
\begin{align}
\Omega_{ab} \gamma_{(2n)}^a \gamma_{(2n)}^b
&= \Omega_{m \bar{l}} \gamma_{(2n)}^m \gamma_{(2n)}^{\bar l}
+ \Omega_{\bar{l} m} \gamma_{(2n)}^{\bar l} \gamma_{(2n)}^m \nonumber \\
&= 2 \Omega_{m \bar{l}} \gamma_{(2n)}^m \gamma_{(2n)}^{\bar l}
-2 \sum_{m=1}^n \Omega_{m \bar{m}},
\end{align}
where we used $\Omega_{m \bar{l}} = - \Omega_{\bar{l} m}$ and
$\{\gamma_{(2n)}^m ,\gamma_{(2n)}^{\bar l}\} = 2\delta_{ml} I_{2^n}$ in the last
step.
As shown in Appendix \ref{Vanishing}, a zero mode $\psi$ can be written as
$\psi = f \ket{+}^{\otimes n}$ with $f$ a section of $L^{\otimes p}$.
Then, by \eqref{gamma matrices identities}, we have
\begin{align}
&D\psi = \im  \bar{w}_m{}^{\bar{\mu}} \gamma_{(2n)}^{\bar{m}}\ket{+}^{\otimes n} \left(\partial_{\bar{\mu}} + p A^L_{\bar \mu} + A^E_{\bar \mu} \right) f = 0 \nonumber \\
& \Rightarrow \quad {}^{\forall} m \in \{1,\cdots, n\}: \quad  \bar{w}_m {}^{\bar{\mu}} \left(\partial_{\bar{\mu}} + p A^L_{\bar \mu} + A^E_{\bar \mu} \right) f = 0 \nonumber \\
& \Rightarrow \quad \left(\partial_{\bar{\mu}} + p A^L_{\bar \mu} + A^E_{\bar \mu} \right) f = 0.
\end{align}
Hence $f$ is a holomorphic section of $L^{\otimes p}$.

\bibliography{ref} 
\bibliographystyle{unsrt} 
\end{document}